\documentclass[12pt, draftclsnofoot, onecolumn]{IEEEtran}
\usepackage{algorithm,amsbsy,amsmath,amssymb,epsfig,bbm,mathrsfs,multirow,amsthm}
\usepackage{color}
\usepackage{algpseudocode,tabularx}
\usepackage{graphicx,caption,subcaption,array,multirow,bm}
\usepackage{makecell}
\usepackage{cite}
\usepackage{threeparttable}

\DeclareMathAlphabet{\mathpzc}{OT1}{pzc}{m}{it}
\DeclareMathOperator*{\argmax}{\arg\!\max}

\renewcommand{\baselinestretch}{1.3725}

\newtheorem{theorem}{\textbf{\textsc{Theorem}}}

\makeatletter
\def\endthebibliography{%
	\def\@noitemerr{\@latex@warning{Empty `thebibliography' environment}}%
	\endlist
}
\makeatother
\makeatletter
\newcommand{\multiline}[1]{%
	\begin{tabularx}{\dimexpr\linewidth-\ALG@thistlm}[t]{@{}X@{}}
		#1
	\end{tabularx}
}
\makeatother
\begin{document}

\title{Joint Coding and Scheduling Optimization for Distributed Learning over Wireless Edge Networks}

\author{Nguyen Van Huynh, Dinh Thai Hoang, Diep N. Nguyen, and Eryk Dutkiewicz
}
\maketitle
\begin{abstract}
Unlike theoretical distributed learning (DL), DL over wireless edge networks faces the inherent dynamics/uncertainty of wireless connections and edge nodes, making DL less efficient or even inapplicable under the highly dynamic wireless edge networks (e.g., using mmW interfaces). This article addresses these problems by leveraging recent advances in coded computing and the deep dueling neural network architecture. By introducing coded structures/redundancy, a distributed learning task can be completed without waiting for straggling nodes. Unlike conventional coded computing that only optimizes the code structure, coded distributed learning over the wireless edge also requires to optimize the selection/scheduling of wireless edge nodes with heterogeneous connections, computing capability, and straggling effects. However, even neglecting the aforementioned dynamics/uncertainty, the resulting joint optimization of coding and scheduling to minimize the distributed learning time turns out to be NP-hard. To tackle this and to account for the dynamics and uncertainty of wireless connections and edge nodes, we reformulate the problem as a Markov Decision Process and then design a novel deep reinforcement learning algorithm that employs the deep dueling neural network architecture to find the jointly optimal coding scheme and the best set of edge nodes for different learning tasks without explicit information about the wireless environment and edge nodes' straggling parameters. Simulations show that the proposed framework reduces the average learning delay in wireless edge computing up to 66\% compared with other DL approaches. The jointly optimal framework in this article is also applicable to any distributed learning scheme with heterogeneous and uncertain computing nodes.
\end{abstract}


\section{Introduction}
\label{Sec:intro}
Recent years have witnessed machine learning as a key enabler to revolutionize the way we communicate, entertain, and work. This is thanks to major machine learning-based breakthroughs in various domains such as computer vision, big data, and natural language processing~\cite{chen2019deep}. However, most of machine learning applications require a massive volume of data to achieve sufficient training accuracy. With the exponential growth in both data volume and data sources, training a given learning model in a centralized manner often requires a large, if not excessive, amount of computation resources and time, especially for high-dimensional training data~\cite{wang2020convergence}. To overcome this problem, distributed learning has been introduced recently~\cite{li2018learning, huang2017deep, Chen2020Convergence,Zhu2019Broadband,Saputra2019Distributed}. In particular, a highly-complex learning task can be partitioned into multiple sub-learning tasks, and then these sub-learning tasks can be transmitted to several edge nodes for executing. In this way, the computation load at the centralized server can be offloaded to multiple edge nodes, and thus improving the training speed. As a result, distributed learning over wireless edge networks finds its applications in various emerging machine learning services that demand low delays such as autonomous vehicle, augmented reality, and virtual reality~\cite{gong2020delay}.

Although distributed learning over wireless edge networks has many advantages and applications in practice, it has been facing some technical challenges. First, it is pointed out that the performance of a distributed system is greatly affected by the straggling problem at the edge computing devices~\cite{lee2017speeding,lee2015speed,prakash2020coded}. In particular, this straggling problem can cause unpredictable computing latency due to several factors such as resource sharing, maintenance activities, power regulations, and hardware configurations~\cite{lee2017speeding}. Consequently, the computing latency is usually determined by the slowest computing edge node. In the worst case, if an edge node is highly-straggling, the learning task will stay in the system for a long time. Consequently, the computing latency of the whole system will be significantly increased. Second, the data privacy protection of the conventional distributed learning is not guaranteed as the edge nodes can derive information from the assigned sub-learning tasks. Moreover, transmitting sub-learning tasks over wireless links may lead to another security concern as an attacker can eavesdrop the transmitted data over the wireless links. These security problems are very serious as private information such as finance data and medical records can be leaked to the third party. Third, distributed learning over wireless edge networks suffers from wireless link failures. Re-transmissions can be performed for failed messages. However, this may significantly increase the training time for the system.

To overcome the aforementioned challenges, the coded computing technique~\cite{lee2017speeding} is introduced as a highly-effective solution. Specifically, the principle of the coded computing is utilizing advanced coding theoretic mechanisms to inject and leverage data/computation redundancy in order to mitigate the effects of the straggling problems as well as to protect the learning tasks' privacy at the edge nodes and over the wireless links~\cite{lee2017speeding, prakash2020coded,dhakal2019coded}. With the coded computing technique, the computation latency is now determined by a group of the fastest edge computing devices~\cite{lee2015speed, lee2017speeding}. In other words, the coded computing technique does not require all assigned edge nodes to send back their computed results as in the traditional distributed edge computing. Similarly, the effects of unstable wireless links can be mitigated as the coded computing mechanism may ignore computed results from edge nodes with unstable wireless links if it has received sufficient computed results from other edge nodes with good wireless connections. Finally, the sub-learning tasks are encoded before sending to the assigned edge nodes, resulting in a high data privacy protection. As a result, applications of coded computing have been widely adopted in distributed learning systems recently~\cite{lee2017speeding,lee2015speed,prakash2020coded, dhakal2019coded,severinson2017block, dutta2019short}. Nevertheless, existing works usually ignore the effects of wireless communications which can lead to serious degradation in the system performance. Moreover, the dynamics and uncertainty of straggling problems at edge nodes and wireless links are also not considered in the literature. In the following, we first discuss current related works using coded computing and then highlight our main contributions in this work.

\subsection{Related Work}
Recently, several works in the literature have been proposed to improve the performance of the coded computing mechanisms for distributed learning systems~\cite{lee2017speeding,severinson2017block,dutta2019short,bitar2020minimizing,tandon2017gradient,karakus2019redundancy,park2018hierarchical,haddadpour2018codes,ferdinand2018hierarchical}. In~\cite{lee2017speeding}, the authors propose a new maximum distance separable (MDS) code design for matrix multiplication which is the most common operation in machine learning algorithms. In particular, the MDS code aims to encode $k$ learning tasks into $n$ coded learning tasks, where $n \geq k$. These encoded tasks are then distributed to $n$ workers to execute. As soon as $k$ workers complete their assigned tasks and send the results to the master node, the master node can decode them to obtain the expected results. In this way, the effect of straggling workers can be significantly mitigated. The authors then demonstrate that with $n$ homogeneous workers, the MDS code can speed up the distributed matrix multiplication by a factor of $\log n$. The authors in~\cite{severinson2017block} then extend the MDS code for large-scale matrix multiplication. Specifically, the key idea is to partition a large-scale matrix into sub-matrices. Then, the MDS code is applied for each sub-matrix. Although the matrix multiplication delay is similar to that of the conventional MDS code~\cite{lee2017speeding}, thanks to shorter MDS codes, the proposed scheme can achieve a lower delay in encoding and decoding compared to that of the conventional MDS code. Similarly, the authors in~\cite{tandon2017gradient} propose a gradient coding method based on the MDS code~\cite{lee2017speeding} for the synchronous gradient descent method. By using this code, the server can obtain the final gradient of any loss function even if a number of workers do not return their gradient results. The experimental results then demonstrate that the proposed gradient coding mechanism can significantly mitigate the straggling problems at workers compared to those of the uncoded schemes. Unlike~\cite{tandon2017gradient}, the authors in~\cite{karakus2019redundancy} propose to encode the dataset with built-in data redundancy for linear regression tasks. At every training step, the missing results from straggling nodes can be compensated by using the structured computing redundancy added by the proposed coding mechanism. Experiments then show that the proposed coding mechanism can significantly reduce the system computing delay.

It is worth noting that aforementioned works and others in the literature mostly focus on optimizing coding mechanisms only. Their application to wireless edge computing is not straightforward due to the inherent uncertainty of wireless channel quality. In particular, when the wireless link between the server and an edge node is disconnected, transmitted data (i.e., sub-learning tasks sent from the server and results sent from the edge node) need to be re-transmitted. This consequently drag out the training time of the whole system. For that, the authors in~\cite{prakash2020coded} introduce an effective coded computing framework for non-linear distributed machine learning, namely CodedFedL, that adds structured coding redundancy to mitigate straggling problems in both edge nodes and wireless links. Specifically, each edge node privately generates a matrix from a probability distribution with mean $0$ and variance $1$. This matrix is then applied on the weighted local dataset to compute a local parity dataset. All local parity datasets of edge nodes are then combined at the server to obtain a global parity dataset. Gradient over the global parity dataset will be used to replace missing gradient updates from straggling edge nodes. The size of the local parity datasets is the coding redundancy. The authors then formulate an optimization problem to find the optimal amount of coding redundancy based on the the conditions of edge nodes and wireless links. Numerical experiments then show that CodedFedL can speed up the training time of federated learning by up to 15 times compared to those of other approaches.

In~\cite{kim2020coded}, the authors introduce an extension of the MDS code considering delay/latency caused by unstable wireless links. In particular, the authors point out that under dynamically changing edge environments, wireless edge nodes may not be able to complete their computations within a given deadline. Thus, the authors propose a new coding mechanism that can incorporate partially-finished computations from edge nodes into the computation recovery at the server. Similarly, the authors in~\cite{ha2019coded} aim to minimize the communication and computing delays by considering both wireless and computing impairments. In particular, the number of edge nodes for executing learning tasks is optimized based on the interference, imperfect channel state information, and straggling processors. Nevertheless, all the aforementioned solutions and others in the literature require complete environment information in advance, which may not be practical to implement. In reality, environment related parameters like link failures and straggling are dynamic and uncertain, especially wireless channel-dependent ones. They can randomly occur at both the edge nodes and wireless links due to unpredictable factors such as maintenance activities, hardware errors, random obstacles, and interference. Without considering these factors, existing solutions may not be able to achieve a highly reliable, efficient and robust performance for distributed learning over wireless edge networks. More importantly, all the current works only optimize the number of edge nodes to execute learning tasks (i.e., optimal values of $(n,k)$ code) and overlook the fact that different edge nodes may have different computing resources, wireless connections, and hardware configurations. As such, selecting the best set of nodes, instead of the number of nodes to execute learning tasks given the current status of the whole system is very critical to further improve the performance of coded distributed learning in wireless edge networks.
\subsection{Main Contributions}
Given the above, this work proposes a jointly optimal coding and scheduling framework for distributed learning over wireless edge networks. In particular, we consider a wireless edge network consisting of a mobile edge server (MEC) connected to various edge nodes with different hardware configurations via different wireless links. When a learning task arrives at the MEC server, it will be encoded into sub-learning tasks by using an MDS-based code\footnote{Note that our proposed solution can not only apply for the MDS code proposed in~\cite{lee2017speeding} but also can apply for other codes. In particular, most of the coded computing techniques aim to optimize the amount of coding redundancy, e.g., \cite{prakash2020coded}, which is similar to $k$ in the MDS code. Therefore, our proposed solution can be straightforwardly extended to other coding techniques.}. Then, these sub-learning tasks are sent to a set of $n$ selected edge nodes to execute. When a predefined number of edge nodes (i.e., $k$ where $k \leq n$) complete their assigned sub-learning tasks, their results can be aggregated to obtain the final result of the original learning task. However, finding an optimal MDS code (i.e., a pair of $n$ and $k$) \emph{and} the best edge nodes (referred to as the optimal scheduling) for each learning task under the dynamic of edge nodes (e.g., available or unavailable) and wireless environment (e.g., good or bad channel condition) is a challenging problem. Solving such a problem in practice is even more difficult as one also needs to account for the uncertainty of wireless links and edge nodes. This is the unpredictable failures or straggling links/devices. To the best of our knowledge, all current works cannot effectively address all these problems.

To tackle the above problem, we first develop a Markov decision process (MDP) framework to capture the aforementioned dynamics and uncertainty of the system such as diverse learning tasks, computing resources, straggling issues at different edge nodes, and wireless channel conditions. To minimize the communication and computing delays, one can rely on the Q-learning algorithm to obtain the optimal coding and edge node scheduling policy. The key idea of this algorithm is learning through interactions with the environment and gradually finds the optimal policy. Nevertheless, the Q-learning algorithm usually takes a long time to converge to the optimal policy, especially for distributed learning systems which usually involve with high-dimensional state and action spaces. Moreover, if the state space is continuous, the conventional Q-learning algorithm may not be able to effectively address the dynamic optimization problem. Therefore, we propose a highly-effective deep reinforcement learning algorithm based on the idea of using the deep dueling neural network architecture~\cite{wang2016dueling} to facilitate the learning process of the distributed learning system. In particular, as the Q-function of each state-action pair is estimated by the deep dueling neural network instead of Q-table as in the conventional Q-learning algorithm, our proposed algorithm can effectively handle the continuous state space. Moreover, different from conventional deep reinforcement learning approaches, this proposed algorithm separately estimates the advantage and value functions for each state-action pair with two streams of hidden layers in the deep dueling neural network architecture~\cite{wang2016dueling}. These two functions are then combined at the output layer to derive the optimal action, i.e., coding and scheduling policy. In this way, the learning process is significantly improved and stable as the unnecessary relations between the values of states and the advantages of corresponding actions are mitigated. For example, selecting MDS codes with high values of $n$, i.e., processing learning tasks in many edge nodes, only benefit when learning task sizes are large. Extensive simulation results show that the proposed solution can jointly obtain the optimal code and the best edge nodes to perform learning tasks given the uncertainty and dynamic of wireless channels and straggling computing at edge nodes. Under the optimal policy, the average latency for learning tasks can be reduced by 66\% compared to those of the conventional coded distributed learning methods. The major contributions of this paper are highlighted as follows:

\begin{itemize}
	\item Propose a highly effective distributed learning framework leveraging outstanding advantages of coded computing as well as abundant computing resources from multiple collaborative edge nodes to securely and effectively execute learning tasks.
	
	\item Propose a jointly optimal coding and scheduling framework for distributed learning over wireless edge networks. Under this framework, one can simultaneously select the optimal code as well as the optimal edge nodes for each learning task given the uncertainty of the edge nodes and wireless links. To the best of our knowledge, our paper is the first work which can jointly optimize both coding and edge node scheduling for coded computing.
	
	\item Develop a highly-effective deep reinforcement learning algorithm for coded computing over wireless edge networks by utilizing the advanced deep dueling neural network architecture~\cite{wang2016dueling} to address the slow-convergence and non-discrete problems of conventional reinforcement learning algorithms (e.g., Q-learning and deep Q-learning algorithms). By separately estimating the advantage and value functions, unnecessary relations between the values of states and the advantages of corresponding actions are mitigated, resulting in a high learning rate. This feature is especially useful as the sever needs not only optimizing the code, but also selecting the best edge nodes to execute learning tasks at the same time.
	
	\item Perform extensive simulations to show the efficiency of our proposed solution compared to those of the conventional approaches (e.g., \cite{lee2017speeding}). Moreover, we discuss and analyze various scenarios to provide insightful designs for distributed learning over wireless edge networks with the coded computing mechanism.
\end{itemize}

The rest of this paper is organized as follows. Section~\ref{Sec.System} presents the system model and the computing and communication models. The MDP framework and the problem formulation are provided in Section~\ref{Sec:prob}. Section~\ref{sec:QDeepQ} presents the deep dueling algorithm in details. Simulation results are discussed in Section~\ref{sec:evaluation}. Finally, conclusions are highlighted in Section~\ref{sec:conclusion}.

\section{System Model}
\label{Sec.System}
\begin{figure}[!]
	\centering
	\includegraphics[scale=0.43]{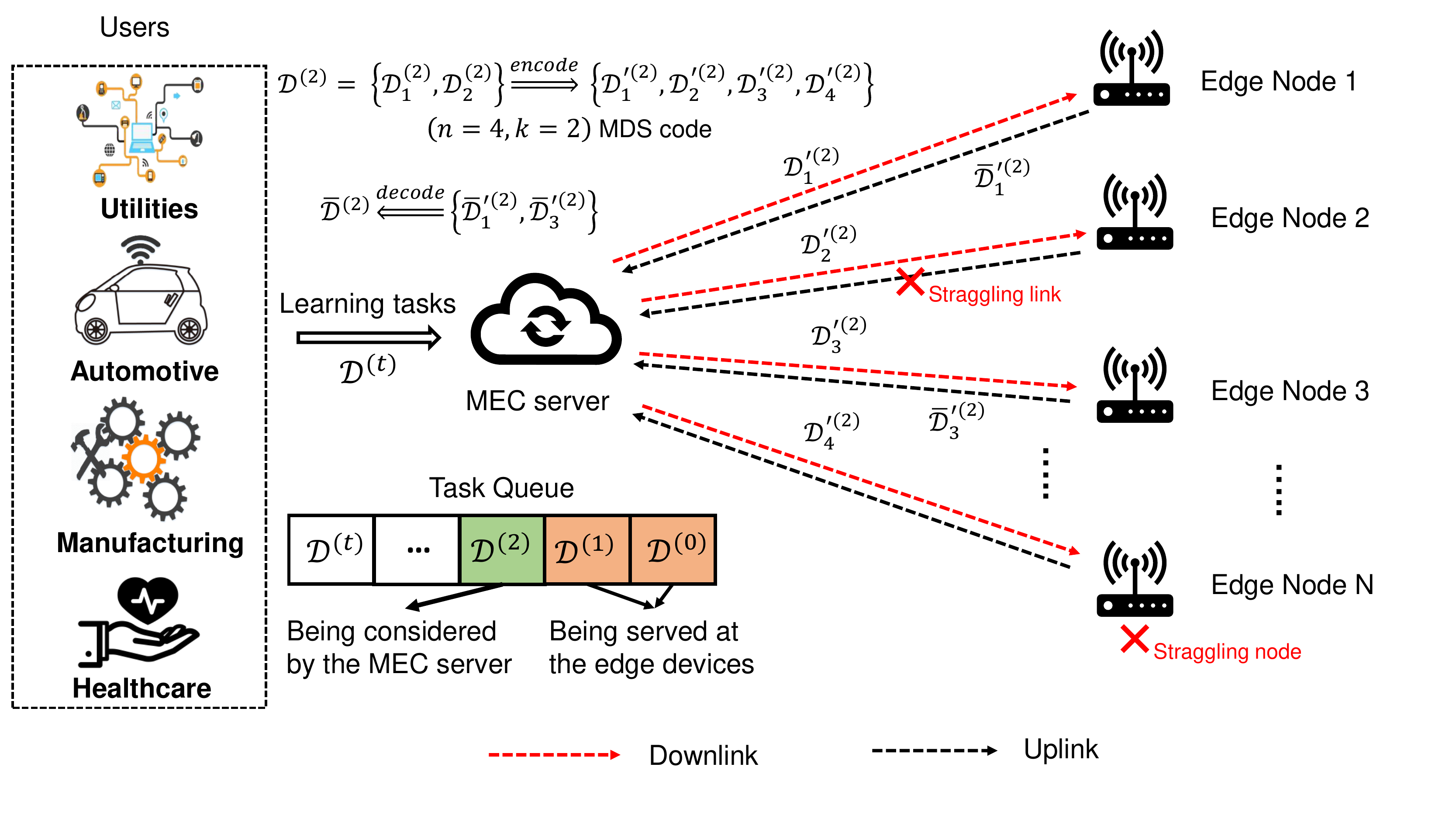}
	\caption{System model for coded distributed learning over wireless edge network. Here, we illustrate the case when learning task $\mathcal{D}^{(2)}$ is processed with $(n=4, k=2)$ MDS code. The sub-learning tasks are sent to edge nodes $1,2,3,\mbox{ and }N$ to process. Then, when edge node $2$ is disconnected, and edge node $N$ is straggling, the learning task $\mathcal{D}^{(2)}$ still can be completed by using computed results from edge nodes $1$ and $3$.}
	\label{Fig.system_model}
\end{figure}

We consider a distributed edge learning system that consists of a mobile edge computing (MEC) server and a set of $N$ edge nodes denoted by $\mathbf{E} = \{E_1, \ldots, E_j, \ldots, E_N\}$. The edge nodes communicate with the MEC through wireless links as illustrated in Fig.~\ref{Fig.system_model}. Let $C_j$ denote the wireless link that connects the MEC server and edge node $E_j$. Practically, different links may be allocated on different channels, and thus their properties, e.g., fading conditions, interference, and disconnection probability, may be different. At the MEC server, learning tasks that arrive to the system are stored in a queue of size $M$. In this paper, we assume that time is slotted. In each time slot, a learning task arrives at the system with probability $\mu$. Note that, our proposed solution still can work well with other packet arrival processes as they will be learned by the proposed algorithm and is not required to be available in advance. At time slot $t$, if a learning task $\mathcal{D}^{(t)}$ arrives at the system and the queue is not full, the learning task will be stored in the queue. Otherwise, the learning task will be dropped. Different learning tasks, e.g., matrix multiplication, data shuffling, or gradient descent for linear regression problems~\cite{lee2017speeding}, may have different data sizes. We denote $f(\mathcal{D}^{(t)})$ as the data size of learning task $\mathcal{D}^{(t)}$.

In our system, learning tasks in the queue are served in a first-come-first-served manner. At each time slot, if the computing resources at the edge nodes are available, the MEC server will look at the queue and consider to serve a learning task which comes earliest in the queue but not yet served by any edge nodes (e.g., $\mathcal{D}^{(2)}$ as illustrated in Fig.~\ref{Fig.system_model}). By using the optimal $(n,k)$ MDS code and the optimal set of edge nodes obtained by our proposed algorithm, this learning task is then encoded into $n$ sub-learning tasks, and these sub-learning tasks are offloaded to edge nodes in the optimal set to execute. These devices then serve the assigned sub-learning tasks and return the results to the MEC server. Note that the learning task still remains in the queue until the MEC server receives $k$ results returned from the edge nodes and successfully decodes them. For example, as illustrated in Fig.~\ref{Fig.system_model}, with $(n=4, k=2)$ MDS code, the MEC server does not need to wait for results from edge node 2 and edge node $N$, which are delayed by the straggling problems. Instead, the MEC server can decode the final result by using computed results returned from edge node 1 and edge node 3 which have better wireless connections and computing power at the considered running time. In contrast, conventional distributed learning models need to wait computed results from all the assigned edge nodes to obtain the final result, and thus dramatically increasing the computing delay of the whole system.

For the ease of notation, we assume that when an edge node receives a sub-learning task, it will use its all computing resource to execute this task. This is also stemmed from the fact that edge nodes (e.g., IoT gateways) are usually equipped with limited resources, and thus they may not be able to serve multiple learning tasks simultaneously. In addition, if an edge node has to process multiple sub-learning tasks at the same time, the straggling problem may be more serious as its computing resources have to share to execute multiple tasks simultaneously. We denote $e_j$ as the state of edge node $E_j$. Specifically, $e_j = 0$ if the edge node is currently busy, i.e., serving one sub-learning task. $e_j = 1$ if the edge node is available, i.e., there is no learning task executing at the edge node. Then, the set of available edge nodes can be denoted as $\mathbf{E}_\mathrm{av} \overset{\text{def}}{=} \{E_j : \forall E_j \in \mathbf{E}$ and $e_j = 1$\}. It is worth noting that our proposed solution can be extended to the case if one edge node can handle multiple tasks at the same time by implementing multiple virtual machines (VMs). Then, each VM can be reserved to execute one learning task. Thus, the MEC just needs to take the available VMs of each edge node into account when it assigns learning tasks to them.
\subsection{Coded Computing for Distributed Learning over Wireless Edge Networks}
The key idea of coded computing techniques is to leverage coding theoretic mechanisms to add structured computing redundancy into learning tasks to mitigate the effects of straggling edge nodes and wireless communication links~\cite{prakash2020coded}. Specifically, the coded computing technique first encodes a learning task into several sub-learning tasks and then sends them to the edge nodes to perform. After some of the edge nodes complete the assigned sub-learning tasks, they will send the results (e.g., trained models or multiplicated matrices) to the MEC server. Here, it is important to note that, the MEC server does not need to wait all sub-learning tasks' results to obtain the final result thanks to outstanding features of the coded computing technique. In this way, the time delay can be significantly reduced, especially in cases if there are some straggling issues occurring at edge nodes and/or at the wireless communication links.

One of the most effective coding techniques used in coded computing is the maximum distance separable (MDS) code~\cite{lee2017speeding}. The fundamentals of the MDS code are illustrated in Fig.~\ref{Fig.system_model}. In particular, with the $(n,k)$ MDS code ($1 \leq k \leq n$), a learning task $\mathcal{D}^{(t)}$ can be first divided into $k$ equal-sized sub-learning tasks $\{\mathcal{D}_1^{(t)}, \mathcal{D}_2^{(t)}, \ldots, \mathcal{D}_k^{(t)}\}$. Then, these sub-learning tasks are encoded by the $(n,k)$ MDS code. After encoding, we get $n$ encoded sub-learning tasks $\{\mathcal{D}_1^{'(t)}, \mathcal{D}_2^{'(t)}, \ldots, \mathcal{D}_n^{'(t)}\}$. These sub-learning tasks are then sent to $n$ edge nodes to execute. Upon receiving any $k$ results from any $k$ edge nodes, the MEC server can decode them to obtain the result. When a learning task is completed, it will be removed from the task queue. After that, the MEC server will inform edge nodes that are still working on the remaining sub-learning tasks to stop performing these sub-learning tasks and make them free.

%

It is worth noting that choosing the values of $n$ and $k$ to maximize the system performance in terms of serving time, delay, and task drop probability is very challenging under the dynamics and uncertainty of the wireless environment as well as straggling problems at the edge computing devices. Currently, an optimal MDS code setting (with optimal values of $n$ and $k$) can be determined based on static optimization methods, e.g., \cite{severinson2017block, tandon2017gradient,prakash2020coded,reisizadeh2019coded}. However, these methods require prior information about the straggling parameters at edge nodes and wireless links. In practice, these parameters usually are not available in advance. Thus, they are not applicable to wireless edge computing as they do not account for the inherent dynamics of wireless channels and edge nodes, leading to uncertainty of straggling problems. Moreover, it is even more challenging when choosing the best edge nodes to execute different learning tasks. To the best of our knowledge, all current existing works cannot address all these problems. Thus, in this work, we propose an intelligent approach which allows the MEC server to dynamically select the optimal MDS code together with the best edge nodes based on current status of the whole system. Note that this approach can not only select the optimal values of $n$ and $k$, but also find the best edge nodes to serve each learning task.
\subsection{Communication and Computation Models}
\label{subsec:timemodel}
Recall that with $n,k$ MDS code, learning task $\mathcal{D}^{(t)}$ is first divided into $k$ equal-sized sub-learning tasks $\{\mathcal{D}_1^{(t)}, \mathcal{D}_2^{(t)}, \ldots, \mathcal{D}_k^{(t)}\}$. These sub-learning tasks are then encoded into $n$ encoded sub-learning tasks $\{\mathcal{D}_1^{'(t)}, \mathcal{D}_2^{'(t)}, \ldots, \mathcal{D}_n^{'(t)}\}$. The encoded sub-learning tasks are finally sent to $n$ edge nodes for processing. In this section, we formulate the serving time for encoded sub-learning task $\mathcal{D}_i^{'(t)}$ ($1 \leq i \leq n$) at a given edge node $E_j \in \mathbf{E}$. For the ease of notation, we denote $T_{\mathrm{serve}}^{(t,i)}$ as the total serving time of sub-learning task $\mathcal{D}_i^{'(t)}$. Thus, $T_{\mathrm{serve}}^{(t,i)}$ can be written as:
\begin{equation}
	\label{eq:time_serve}
	T_{\mathrm{serve}}^{(t,i)} = T_{\mathrm{se}}^{(t,i)} + T_{\mathrm{cmp}}^{(t,i)} + T_{\mathrm{es}}^{(t,i)},
\end{equation}
where $T_{\mathrm{se}}^{(t,i)}$ and $T_{\mathrm{es}}^{(t,i)}$ are the communication time for sending $\mathcal{D}_i^{'(t)}$ from the MEC server to edge node $E_j$ and the time for sending back its computed result from edge node $E_j$ to the MEC server through wireless link $C_j$, respectively. Note that both the uplink (from edge node $E_j$ to the MEC server) and the downlink (from MEC server to edge node $E_j$) can share the same channel as the MEC server and edge node $E_j$ do not need to transmit data at the same time. $T_{\mathrm{cmp}}^{(t,i)}$ is the time that edge node $E_j$ requires to complete the sub-learning task. With the high-speed backhaul connections from the edge node to the server (e.g., via mmWave), a sub-learning task or its result can be transmitted over the wireless link within one time slot. To capture the dynamics of the wireless link $C_j$ from the edge node $E_j$ to the server, e.g., due to fading or moving obstacles and/or interfere with surrounding RF signals, let's define $p_j$ as the disconnection probability of the wireless link from the MEC server to edge node $E_j$ over a given time slot. We then denote $\mathbf{p}=\{p_1, \ldots, p_j, \ldots, p_N\}$ as the set of disconnection probabilities corresponding to wireless channels $\{C_1, \ldots, C_j, \ldots, C_N\}$. Note that our proposed solution does not require these probabilities in advance (unlike most of existing works~\cite{ha2019coded,Li2020Multi} using coded computing). Instead, it learns these probabilities by interacting with the environment.

In the case wireless link $C_j$ is disconnected, the MEC server or edge node $E_j$ needs to resend its data in the next time slot. As such, the communication delay increases, especially when the disconnection probability is high. We thus can formulate $T_{\mathrm{se}}^{(t,i)}$ and $T_{\mathrm{es}}^{(t,i)}$ as follows:
\begin{equation}
	\label{eq:commu_time}
	T_{\mathrm{se}}^{(t,i)} = T_{\mathrm{es}}^{(t,i)} = H_j\xi,
\end{equation}
where $\xi$ is the duration of a time slot and $H_j$ is the number of time slots needed to successfully transmit data on wireless channel $C_j$. $H_j$ is i.i.d based on the \textit{Geometric} distribution with the successful probability $p_\mathrm{success} = 1- p_j$ as follows~\cite{prakash2020coded}:
\begin{equation}
	Pr(H_j = x) = p_j^{x-1}(1-p_j), x = 1, 2, 3,\ldots
\end{equation}
According to the \textit{Geometric} distribution, a higher value of $p_j$, i.e., disconnection probability, results in a higher value of $H_j$. Thus, in many scenarios, the delay caused by unstable connections is even more serious than that of straggling devices, especially when the link disconnection probability is very high~\cite{Wu2020Latency}. To deal with this issue, in the sequel, we propose an effective learning solution that can learn the disconnection probabilities of all wireless links to avoid bad connections when serving learning tasks (e.g., assigning tasks to edge nodes with more favorable connections). Hence, the straggling effects caused by unstable links can be significantly mitigated.

The computing time $T_{\mathrm{cmp}}^{(t,i)}$ of sub-learning task $\mathcal{D}_i^{'(t)}$ ($1 \leq i \leq n$) at edge node $E_j$ is the sum of the stochastic time for random memory access during read/write cycles and the deterministic time for processing data~\cite{lee2017speeding,prakash2020coded,zhang2019model}. Thus, $T_{\mathrm{cmp}}^{(t,i)}$ can be written as follows:
\begin{equation}
	\label{eq:compu_time}
	T_{\mathrm{cmp}}^{(t,i)} = \underbrace{f(\lambda_j)}_{\text{stochastic time}} + \underbrace{\eta_j |\mathcal{D}_i^{'(t)}|}_{\text{deterministic time}},
\end{equation}
where $f(\lambda_j)$ is a random variable denoting the stochastic component of the computing time caused by the straggling problem at edge node $E_j$. $f(\lambda_j)$ follows an exponential distribution with rate $\lambda_j$~\cite{zhang2019model, prakash2020coded}, i.e., $p_{f(\lambda_j)}(x) = \lambda_je^{-\lambda_jx}, x \geq 0$. $\eta_j$ is the deterministic time for edge node $E_j$ to process one data point (e.g., one row in the matrix). $|\mathcal{D}_i^{'(t)}|$ is the data size of sub-learning task $\mathcal{D}_i^{'(t)}$. We denote $\bm{\lambda}=\{\lambda_1, \ldots, \lambda_j, \ldots, \lambda_N\}$ as the set of rate parameters determining the stochastic time at edge nodes. Specifically, $\frac{1}{\lambda_j}$ is the average stochastic time that can be considered as straggling parameter of edge node $E_j$. The straggling parameter depends on many factors such as shared resources, maintenance activities, power limitation, and random memory access~\cite{dutta2019short, kim2020coded}. With high straggling parameters, edge nodes need more time for processing learning tasks. Therefore, avoiding straggling edge nodes is crucial in distributed learning as they can significantly slow down the learning process. In practice, the straggling problems at edge nodes may occur randomly and cannot be effectively predicted. To tackle this problem, our framework below aims to learn the straggling parameters of edge nodes to find the best edge nodes for each learning task, and thus remarkably mitigating the impacts of straggling devices.
\subsection{Learning-Task Delay Minimization Problem}
From (\ref{eq:time_serve}), (\ref{eq:commu_time}), and (\ref{eq:compu_time}), the total serving time of a sub-learning task $\mathcal{D}_i^{'(t)}$ can be expressed as:
\begin{equation}
	\label{eq:serving_time_subtask}
	T_{\mathrm{serve}}^{(t,i)} = \Big(2H_j\xi + f(\lambda_j) + \eta_j |\mathcal{D}_i^{'(t)}|\Big) c_{i,j}, \forall i \in \{1, \ldots, n\}, \forall j \in \{1,2, \ldots, N\}
\end{equation}
where $c_{i,j}$ is a scheduling binary decision. $c_{i,j} = 1$ if sub-learning task $\mathcal{D}_i^{'(t)}$ is served by edge node $E_j$, and $c_{i,j} =0$, otherwise. Note that each sub-learning task is only severed by one edge node, and thus $\sum_{j=1}^{N}c_{i,j} = 1, \forall i \in \{1, \ldots, n\}$. With the $(n,k)$ MDS code, the MEC server only needs $k$ results from any $k$ (out of $n$) edge nodes to successfully decode the result for learning task $\mathcal{D}^{(t)}$. Thus, the total serving time for a learning task $\mathcal{D}^{(t)}$ is the serving time of the $k$-th completed sub-learning task, which can be expressed as follows:
\begin{equation}
	\label{eq:serving_time_task}
	T_{\mathrm{serve}}^{(t)} = \min_{k-th}\Bigg(\Bigl\{T_{\mathrm{serve}}^{(t,i)}: \forall i \in \{1, 2, \ldots, n\}\Bigr\}\Bigg),
\end{equation}
where $\displaystyle\min_{k-th}(.)$ returns the $k$-th minimum value of a set. For example, $\displaystyle\min_{3-th}\Big(\{1, 5, 10, 4, 6\}\Big) = 5$. We then can formulate the serving time minimization problem for a learning task $\mathcal{D}^{(t)}$ as follows:
\begin{eqnarray}
	\label{eq:minimization}
	\min_{n, k, \{c_{i,j}\}}	& &	T_{\mathrm{serve}}^{(t)},	\\
	\mbox{s.t.}			& & 1 \leq k \leq n, \forall n \in \{1, 2, \ldots, |\mathbf{E}_\mathrm{av}|\},\nonumber\\
	& & c_{i,j} \in \{0,1\}, \forall i \in \{1, 2, \ldots, n\} \mbox{ and } \forall j \in \{1, 2, \ldots, N\},\nonumber\\
	&& \sum_{j=1}^{N}c_{i,j} = 1, \forall i \in \{1, \ldots, n\} \mbox{ and } \forall j \in \{1, \ldots, N\},\nonumber\\
	&& c_{i,j} = 1, \mbox{ if } e_j = 1, \forall j \in \{1, \ldots, N\}.\nonumber
\end{eqnarray}
In Theorem~\ref{theo:nphard}, we show that the delay minimization problem in~(\ref{eq:minimization}) is an NP-hard problem.

\begin{theorem}
	\label{theo:nphard}
	The joint coding and scheduling optimization problem~(\ref{eq:minimization}) is NP-hard.
\end{theorem}
\begin{proof}
	It can be observed that the optimization problem in (\ref{eq:minimization}) is a form of the Knapsack problem~\cite{Knapsack}. In particular, (\ref{eq:minimization}) aims to find the optimal MDS code (i.e., optimal values of $n$ and $k$) and the optimal scheduling policy (i.e., the best set of $\{c_{i,j}\}$) to minimize the serving time for each learning task. It is worth noting that the problem in (\ref{eq:minimization}) is much more complicated than the Knapsack problem as the serving time of each edge node is a stochastic value and changes over time as shown in (\ref{eq:serving_time_subtask}). As a result, solving (\ref{eq:minimization}) is more difficult than solving the Knapsack problem. As shown in~\cite{Knapsack}, the Knapsack problem is an NP-hard problem. As such, the optimization problem in~(\ref{eq:minimization}) is also an NP-hard problem.
\end{proof}

As shown in Theorem~\ref{theo:nphard}, minimizing the serving time of each learning task is NP-hard, even if the MEC server knows the environment parameters such as $p_j$, $\lambda_j$, and $\eta_j$ in advance. However, in practice, these parameters are usually not available in advance due to the dynamics and uncertainty of edge nodes and wireless links. Moreover, this paper aims to minimizing the average delay for all learning tasks, which is much more challenging than that for a single learning task as in~(\ref{eq:minimization}). The reason is that learning tasks are sharing the same computing resources from edge nodes, and thus the optimal coding and scheduling for a learning task will have significant impacts on all next arrival learning tasks. Thus, all current static optimization tools in existing works in the literature~\cite{lee2015speed, lee2017speeding, prakash2020coded} may not be effective in dealing with these practical issues. To tackle this and to account for the dynamics and uncertainty of wireless connections and edge nodes, we reformulate the problem as a Markov decision process and then design a novel deep reinforcement learning algorithm that employs the deep dueling neural network architecture to find the jointly optimal coding and scheduling policy for different learning tasks without explicit information about the wireless environment and edge nodes' straggling parameters.
\section{Coded Computing for Distributed Learning Formulation}
\label{Sec:prob}
We first adopt the Markov decision process (MDP) framework to formulate the system delay minimization problem for distributed learning over wireless edge networks. Generally, the MDP is defined by a tuple $<\mathcal{S}, \mathcal{A}, r>$ where $\mathcal{S}$ is the state space, $\mathcal{A}$ is the action space, and $r$ is the immediate reward function of the system. With the MDP framework, the MEC server can dynamically make optimal actions, i.e., select optimal MDS codes as well as the best edge nodes for serving sub-learning tasks, based on its current states, e.g., task queue and available edge nodes' resources, to maximize its long-term average reward. Thus, this framework can significantly reduce the average delay of learning tasks.
\subsection{State Space}
As stated above, a learning task stored in the task queue at the MEC server will be severed in the first-come-first-served manner. To serve a learning task, the MEC first selects $n$ available edge nodes and chooses $(n,k)$ MDS code to encode the learning task. After that, the encoded sub-learning tasks are transmitted to a set of optimal edge nodes. The learning task will not be removed from the queue until the MEC server receives any $k$ results from the edge nodes. Thus, the queue size, the available edge nodes, and the size of the considered learning task are important factors that should be captured by the system state $s$. We then define the state space $\mathcal{S}$ of the system as follows:
\begin{equation}
\label{eq:state}
\begin{aligned}
\mathcal{S} \triangleq \Big\{\!\!(m, f, \{e_1,\ldots, e_j, \ldots, e_N\}): m \in \{0, \ldots, M\}; f \geq 0; e_j \in \{0,1\}, \forall j \in \{1, \ldots, N\}\!\!\Big\},
\end{aligned}
\end{equation}
where $m$ represents the number of learning tasks currently waiting in the queue, $f$ is the size of the current considered learning task. Note that the current task size equals 0 only when: (i) the task queue is empty or (ii) all current tasks in the queue are being served and there is no new task arriving. The system state is then defined as a composite variable $s = (q,f,\{e_1, \ldots, e_j, \ldots, e_N\}) \in \mathcal{S}$. Note that the environment parameters such as straggling parameters of edge nodes and wireless links as well as the channel quality cannot obtain in advance as discussed in the previous sections. Thus, our system state space does not take these parameters into account. Instead, these parameters are implicitly captured in the immediate function defined below and then learned by our proposed learning algorithm to simultaneously obtain the optimal coding and scheduling policy.

\subsection{Action Space}
As mentioned, most of current works only focus on optimizing the optimal code, i.e., the optimal values of $n$ and $k$~\cite{lee2017speeding}. However, the straggling problems at wireless links and edge nodes are randomly and unpredictable. As such, optimizing only the values of $n$ and $k$ often leads to sub-optimal solutions in terms of the system delay. To address this problem, our work aims to find not only the optimal code but also the best set of edge nodes to serve a learning task based on the current system state $s$. Thus, an action is combination of $n$, $k$, and the set of edge nodes to serve the current learning task. At state $s$ defined in (\ref{eq:state}), denote $\mathbf{E}_\mathrm{av}$ as the set of all available edge nodes ($\mathbf{E}_\mathrm{av} \subseteq \mathbf{E}$), we have the action space of the system as follows:
\begin{equation}
	\mathcal{A} \triangleq \{a_s\} = \{(0,0,\emptyset), (n,k,\mathbf{E}_\mathrm{b})\}, \forall n \in \{1, \ldots, |\mathbf{E}_\mathrm{av}|\}, \forall k \in \{1, \ldots, n\}, \forall \mathbf{E}_\mathrm{b} \in \binom{\mathbf{E}_\mathrm{av}}{n},
\end{equation}
where $a_s$ is the action selected at state $s$ and $\mathbf{E}_\mathrm{b}$ is the set of $n$-best edge nodes to serve the current learning task. $|\mathbf{E}_\mathrm{av}|$ is the total number of available edge nodes at state $s$. $\binom{\mathbf{E}_\mathrm{av}}{n}$ is the combination operation that returns all size-$n$ subsets of $\mathbf{E}_\mathrm{av}$. For example, if $\mathbf{E}_\mathrm{av} = \{E_1, E_2, E_3\}$ and $n=2$, we have $\binom{\mathbf{E}_\mathrm{av}}{2} = \Big\{\{E_1, E_2\}, \{E_1, E_3\}, \{E_2, E_3\}\Big\}$. From this set, the MEC server can select any size-$2$ subset of edge nodes to serve the learning task, i.e., $\mathbf{E}_\mathrm{b} = \{E_1, E_2\} \mbox{, } \mathbf{E}_\mathrm{b} = \{E_1, E_3\} \mbox{ or } \mathbf{E}_\mathrm{b} = \{E_2, E_3\}$. In general, $a_s = (n,k,\mathbf{E}_\mathrm{b})$ if the MEC server chooses $(n,k)$ MDS code to encode the learning task and the edge nodes in set $\mathbf{E}_\mathrm{b}$ to execute the encoded sub-learning tasks. $a_s = (0,0, \emptyset)$ if the MEC server chooses to stay idle, i.e., not select any code nor edge nodes to execute the task but wait until the next time slot.
\subsection{Immediate Reward}
In this work, we aim to minimize the average long-term delay of learning tasks. In general, the delay of a learning task is defined as the time it spends in the system, including the waiting/queuing time and the serving time. However, in our system, a learning task will remain in the queue until the MEC server successfully decodes the results sent back from the assigned edge nodes. Thus, the average delay of a learning task can be calculated as the time it waits in the task queue from the time it arrives at the system until the MEC server successfully decodes its result. Note that at time slot $t$ when the MEC server makes action $a_t$ to serve a learning task at state $s_t$, it may not be able to know when the learning task is completed. The main reason is that the time to compete this task is determined by the communication time and the computing time as expressed in (\ref{eq:commu_time}) and (\ref{eq:compu_time}), respectively. However, the communication time and the computing time is not deterministic due to the random straggling problems in both the edge nodes and wireless links. As a result, to determine the immediate reward when an action is made, we can observe the number learning tasks in the queue. The reason is that we can implicitly capture the delay of all learning tasks through the length/size of the task queue according to the Little theorem. Thus, we define an immediate reward function for action $a_t$ at state $s_t$ using the instantaneous size of the queue as follows:
\begin{equation}
	\label{eq:reward}
	r_t(s_t,a_t) = -m,
\end{equation}
where $m \in \{0, \ldots, M\}$ is the number of learning tasks waiting in the queue after performing action $a_t$ at state $s_t$. By maximizing the immediate reward function, we can minimize the number of learning tasks in the queue, and thus minimizing the average latency of the whole system.
\subsection{Long-Term Delay Minimization Formulation}
This work aims to obtain the optimal coding and scheduling policy which is a mapping from a given state $s$ to an optimal action to minimize the average long-term reward of the system. In other words, we aim to minimize the average number of learning tasks waiting in the queue. Thus, the optimal coding and scheduling policy can be denoted by $\pi^*:\mathcal{S} \rightarrow \mathcal{A}$, which is then expressed as follows:
\begin{eqnarray}
	\label{eq:average_reward}
	\max_\pi	& &	{\mathcal{R}}(\pi)	=	\lim_{T \rightarrow \infty} \frac{1}{T} \sum_{t=1}^{T} {\mathbb{E}} \left( r_t (s_t, \pi(s_t)) \right),	\label{eq:cmdp_obj}
\end{eqnarray}
where $r_t (s_t, \pi(s_t))$ is the immediate reward under policy $\pi$ at time step $t$ and ${\mathcal{R}}(\pi)$ is the average long-term reward of the system under policy $\pi$. To guarantee that the optimal coding and scheduling policy exists and can be obtained with any initial states, we prove that the average reward $\mathcal{R}(\pi)$ is well defined and does not depend on the initial state as stated in Theorem~\ref{theo:limitexists}.

\begin{theorem}
	\label{theo:limitexists}
	The average reward does not depend on the initial state and is well defined.
\end{theorem}
The proof of Theorem~\ref{theo:limitexists} is provided in Appendix~\ref{appendix:limitexist}.
\section{Optimal Coded Edge Computing with Reinforcement Learning Algorithms}
\label{sec:QDeepQ}
\subsection{Q-Learning based Coded Edge Computing}
There are several approaches to solve the dynamic optimization problem in~(\ref{eq:average_reward}) such as value iteration, policy iteration, and linear programming~\cite{filar2012competitive}. However, most of them require complete environment information, e.g., wireless link disconnection probability and straggling parameters of edge nodes, which may not be always available in advance in practice. Thus, in this section, we propose a reinforcement learning algorithm based on the Q-learning algorithm~\cite{watkins1992q} to deal with the dynamics and uncertainty of the environment by learning from previous observations. In particular, the algorithm implements a Q-table to store the Q-values for all the state-action pairs. At a given state $s_t$, the algorithm makes an action $a_t$ based on the $\epsilon$-greedy policy. Specifically, the algorithm selects an action that maximizes Q-value in the Q-table at state $s$ with probability $1-\epsilon$ and chooses a random action at state $s$ with probability $\epsilon$. After that, the algorithm receives an immediate reward $r_t$, and the system then moves to a next state $s_{t+1}$. From these observations, the algorithm updates the Q-value for the state-action pair $(s_t, a_t)$ as follows~\cite{watkins1992q}:
\begin{equation}
\label{Eq:qfunction}
\begin{aligned}
\mathcal{Q}_{t+1}(s_t,a_t) = \mathcal{Q}_t(s_t,a_t) + \tau_t \Big [ r_t(s_t, a_t) + \gamma\max_{a_{t+1}} \mathcal{Q}_t(s_{t+1}, a_{t+1})- \mathcal{Q}_t(s_t,a_t)\Big ],
\end{aligned}
\end{equation}
where $r_t(s_t, a_t)$ is the immediate reward obtained after executing action $a_t$ at state $s_t$, $\gamma \in [0,1)$ is the discount factor that determines the weight of future reward. Typically, with small value of $\gamma$, the algorithm tends to maximize the short-term reward. In contrast, if $\gamma$ approaches 1, the algorithm will select actions to maximize the long-term reward. To find the optimal coding and scheduling policy effectively, $\gamma$ is usually set at high values, e.g., $0.9$. $\tau_t$ is the learning rate that is used to determine the impact of new information to the existing value. In practice, $\tau_t$ is set close to zero, e.g., $0.1$.

It is worth noting that the goal of (\ref{Eq:qfunction}) is finding the temporal difference between the target Q-value $r_t(s_t, a_t) + \gamma \max_{a_{t+1}} \mathcal{Q}_t(s_{t+1}, a_{t+1})$ and the current estimated Q-value $\mathcal{Q}_t(s_t,a_t)$. Through updating the Q-table based on (\ref{Eq:qfunction}), the Q-learning algorithm can gradually converge to the optimal coding and scheduling policy. To guarantee the convergence for the Q-learning algorithm, the learning rate $\tau_t$ is deterministic, non-negative, and satisfies (\ref{Eq:rules})~\cite{watkins1992q}.
\begin{equation}
	\label{Eq:rules}
	\tau_t \in [0,1), \sum_{t=1}^{\infty}\tau_t = \infty, \mbox{ and } \sum_{t=1}^{\infty} ( \tau_t  )^{2} < \infty.
\end{equation}
Under the conditions in (\ref{Eq:rules}), we can show that the Q-learning algorithm will converge to the optimal policy with probability one in the following theorem.
\begin{theorem}
	\label{theo:convergeQ}
	Under~(\ref{Eq:rules}), Q-learning algorithm is guaranteed to converge to the optimal policy.
\end{theorem}
The proof of Theorem~\ref{theo:convergeQ} can be found in~\cite{watkins1992q}.

\begin{algorithm}[t]
	\caption{Optimal Coding and Scheduling Policy with Q-learning Algorithm}
	\label{alg:qlearning}
	\begin{algorithmic}
		\State \textbf{Inputs:} For each state-action pair $(s, a)$, initialize the table entry $\mathcal{Q}(s, a)$ arbitrarily. Observe the current state $s$, initialize values for the learning rate $\tau$ and the discount factor $\gamma$.
		\For{\textit{t=1 to T}}
		\State \multiline{From the current state-action pair $(s_t, a_t)$, execute action $a_t$ based on the $\epsilon$-greedy method and obtain the immediate reward $r_t$ and new state $s_{t+1}$.}
		\State \multiline{Find the maximum value of $\mathcal{Q}_t(s_{t+1}, a_{t+1})$ and then update the table entry for $\mathcal{Q}(s_t, a_t)$ as follows:}
		\begin{equation}
		\begin{aligned}
		\mathcal{Q}_{t+1}(s_t,a_t) = \mathcal{Q}_t(s_t,a_t) + \tau_t\Big [ r_t(s_t, a_t) + \gamma\max_{a_{t+1}} \mathcal{Q}_t(s_{t+1}, a_{t+1})- \mathcal{Q}_t(s_t,a_t)\Big].
		\end{aligned}
		\end{equation}
		\State Replace $s_t \leftarrow s_{t+1}$.
		\EndFor
		\State {\textbf{Outputs:}} $\pi^*(s) = \arg\max_{a} \mathcal{Q}^*(s,a)$.
	\end{algorithmic}
\end{algorithm}

Algorithm~\ref{alg:qlearning} presents the details of the Q-learning algorithm. After a finite steps, the algorithm can obtain the optimal coding and scheduling policy for the system~\cite{watkins1992q}. Nevertheless, the Q-learning based algorithms face the slow-convergence problem, especially with large state and action spaces in our system. Moreover, as conventional Q-learning algorithms only can effectively handle discrete state space, they may not be feasible for our consider problem as the sizes of learning tasks are continuous variables. Therefore, in the following, we propose a highly-effective reinforcement learning algorithm using recent advances of deep neural network and deep dueling architecture in order to not only address the continuous state space problem, but also quickly find the optimal policy for the MEC server.
\subsection{Coded Computing with Deep Dueling Algorithm}
\label{Sec:deepdueling}
The key idea of the deep dueling algorithm is training a deep dueling neural network to find the approximated values of $\mathcal{Q}^*(s,a)$ instead of implementing the Q-table as in the conventional Q-learning algorithms. In particular, in each training iteration, given a current state $s_t$, similar to the Q-learning algorithm, the deep dueling algorithm selects an action $a_t$ based on the $\epsilon$-greedy policy. After performing the action, the algorithm observes reward $r_t$ and next state $s_{t+1}$. In this work, we adopt the experience replay mechanism to improve the efficiency of the training process. Specifically, all transitions $(s_t, a_t, r_t, s_{t+1})$ (i.e., experiences) are stored in a replay memory $\mathbf{D}$. The algorithm then randomly chooses a number of samples from the replay memory and feeds them to the deep neural network for training. As such, the previous experiences can be efficiently learned many times to improve the stability of the algorithm~\cite{mnih2015human}.

\subsubsection{Deep Dueling Neural Network Architecture}
\begin{figure}[!]
	\centering
	\includegraphics[scale=0.17]{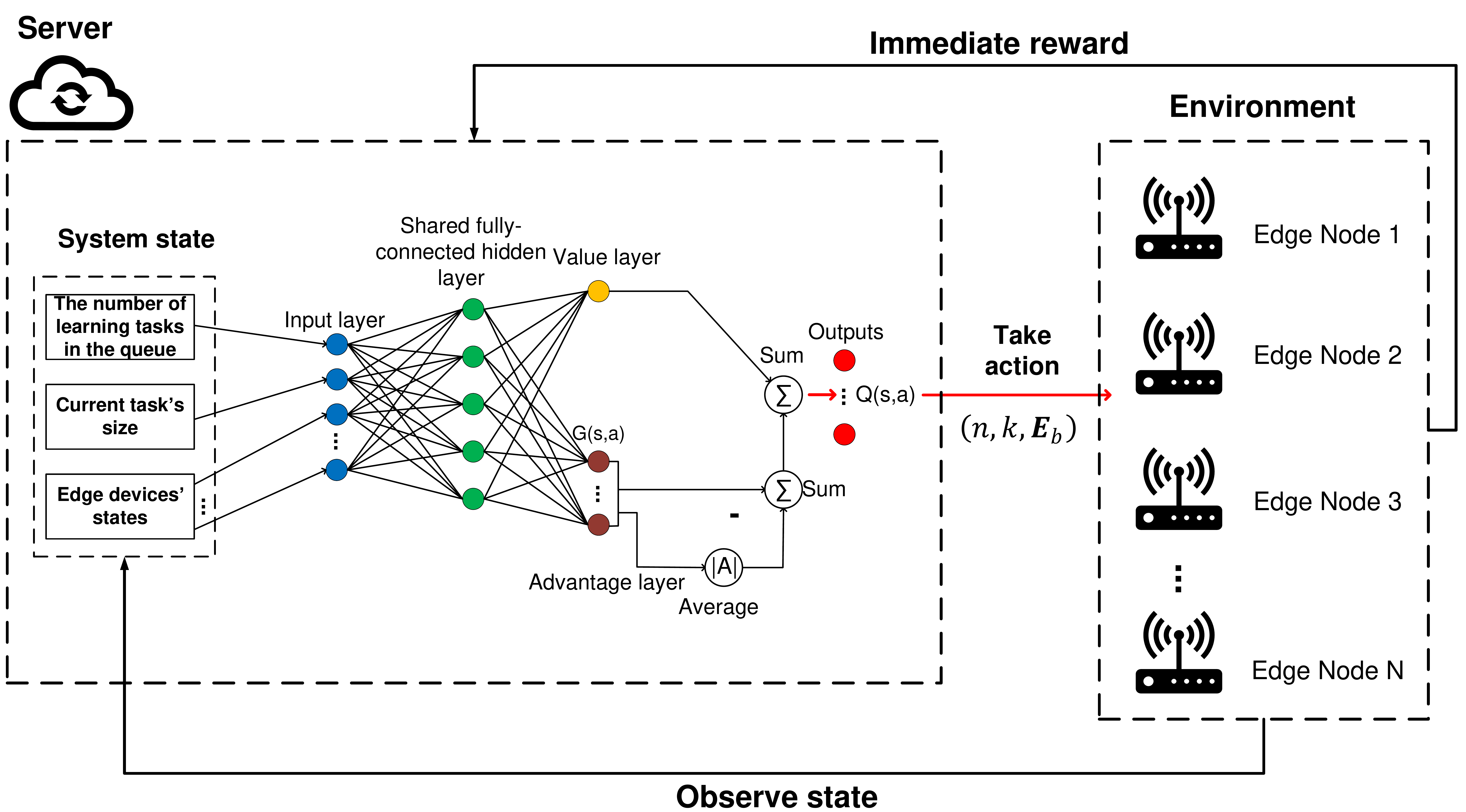}
	\caption{Deep dueling network architecture for coded computing over wireless edge networks.}
	\label{Fig.deepduelingqlearning}
\end{figure}

As mentioned, in this work, we aim to deal with the unpredictable straggling problems occurred on the wireless links and the edge nodes. Thus, conventional deep neural networks may not be able to learn effectively in this highly-dynamic system. Moreover, in this paper, we not only find the optimal code but also select the best edge nodes to serve each learning task. As a result, the action space is high-dimensional, which may degrade the learning rate of conventional deep neural networks. To address these problems, we propose a novel neural network architecture based on the dueling neural network with two streams of hidden layers~\cite{wang2016dueling}. The dueling neural network has demonstrated its effectiveness in many applications~\cite{van2019jam, van2019optimal, van2020deep}. The key idea of this architecture is that in many states, the choices of corresponding actions have no effect on the system~\cite{wang2016dueling}. For example, the MDS codes with high values of $n$ and $k$ only matter when learning task sizes are large. Moreover, when the available resources of the system are limited, choosing different MDS codes makes insignificant differences. As such, instead of estimating the Q-value function for each pair of state and action, we divide it into value and advantage functions. The value function is used to estimate how good it is when the system is at a given state. The advantage function represents the importance of a certain action compared to other actions. Thus, we implement a deep dueling neural network consisting of two streams of fully connected layers to separately estimate the value and advantage functions. These two functions are then combined at the output layer as illustrated in Fig.~\ref{Fig.deepduelingqlearning}. As a result, the deep dueling algorithm can achieve more robust estimates of state values, thereby improving the performance in terms of convergence rate and stability. In the following, we present details of of how to separate the Q-value function into the advantage and value functions.

With a given policy $\pi$, the values of each state-action pair $(s,a)$ and state $s$ are expressed as $\mathcal{Q}^{\pi}(s,a)= \mathbb{E} \big[r_t|s_t=s, a_{t}=a,\pi\big]$ and $\mathcal{V}^{\pi}(s) = \mathbb{E}_{a \sim \pi(s)}\big[\mathcal{Q}^{\pi}(s,a)\big]$, respectively. The advantage function of actions is formulated as $\mathcal{G}^{\pi}(s,a) =  \mathcal{Q}^{\pi}(s,a) - \mathcal{V}^{\pi}(s).$ The values of $\mathcal{V}$ and $\mathcal{G}$ functions are then estimated by the deep dueling neural network. Specifically, one stream of fully-connected layers procedures a scalar $\mathcal{V}(s;\bm{\beta})$, while another one estimates an $|\mathcal{A}|$-dimensional vector $\mathcal{G}(s, a;\bm{\alpha})$, where $\bm{\alpha}$ and $\bm{\beta}$ are the parameters of the deep dueling neural network. Then, these two streams are combined at the output layer to derive the Q-value function as follows:
\begin{equation}
	\label{combined}
	\mathcal{Q}(s, a;\bm{\alpha}, \bm{\beta}) = \mathcal{V}(s;\bm{\beta}) + \mathcal{G}(s, a;\bm{\alpha}).
\end{equation}
Note that $\mathcal{Q}(s, a;\bm{\alpha}, \bm{\beta})$ is a parameterized estimate of the true Q-function. Given $\mathcal{Q}$, it is impossible to derive $\mathcal{V}$ and $\mathcal{G}$ uniquely. This is due to the fact that the Q-value is not changed when subtracting a constant from $\mathcal{G}(s, a;\bm{\alpha})$ and adding the same constant to $\mathcal{V}(s;\bm{\beta})$. This leads to a poor performance as (\ref{combined}) is unidentifiable. To deal with this issue, the Q-value function can be obtained by the following mapping:
\begin{equation}
	\label{ouput_max}
	\mathcal{Q}(s,a;\bm{\alpha}, \bm{\beta}) = \mathcal{V}(s;\bm{\beta}) + \big(\mathcal{G}(s,a;\bm{\alpha})-\max_{a \in \mathcal{A}}\mathcal{G}(s,a;\bm{\alpha})\big).
\end{equation}
As such, the advantage function estimator has no advantage when choosing actions. Intuitively, given $a^*=\argmax_{a \in \mathcal{A}} \mathcal{Q}(s,a;\bm{\alpha}, \bm{\beta})=\argmax_{a \in \mathcal{A}}\mathcal{G}(s,a;\bm{\alpha})$, we have $\mathcal{Q}(s, a^*;\bm{\alpha}, \bm{\beta})=\mathcal{V}(s;\bm{\beta})$. Thus, we can convert (\ref{ouput_max}) to a simple form as follows:
\begin{equation}
	\label{output_average}
	\mathcal{Q}(s,a;\bm{\alpha}, \bm{\beta}) \!\!=\!\! \mathcal{V}(s;\bm{\beta}) \!+\! \big(\mathcal{G}(s,a;\bm{\alpha}) \!-\! \frac{1}{|\mathcal{A}|}\sum_{a}^{}\mathcal{G}(s, a;\bm{\alpha})\big).
\end{equation}
\begin{algorithm}[t]
	\caption{Optimal Coding and Scheduling with Deep Dueling Neural Network Architecture}
	\label{deepduelingqlearning}
	\begin{algorithmic}[1]
		\State Initialize replay memory $\mathbf{D}$ to capacity $\mathcal{D}$.
		\State Build the $\mathcal{Q}$ network with two fully-connected layers with random weights $\bm{\alpha}$ and $\bm{\beta}$.
		\State Build the target $\hat{\mathcal{Q}}$ network with random weights $\bm{\alpha}^- = \bm{\alpha}$ and $\bm{\beta}^- = \bm{\beta}$.
		\For{\textit{iteration=1 to I}}
		\State \multiline{Select action $a_t$ based on the $\epsilon$-greedy policy.}
		\State Execute action $a_t$ and observe immediate reward $r_t$ and next state $s_{t+1}$.
		\State Store experiences $(s_t, a_t, r_t, s_{t+1})$ in $\mathbf{D}$.
		\State Sample random mini-batch of transitions $(s_j, a_{j}, r_j, s_{j+1})$ from $\mathbf{D}$.
		\State Run a gradient descent step on $(y_j-\mathcal{Q}(s_j, a_{j}; \bm{\alpha}, \bm{\beta}))^2$.
		\State Every $C$ steps reset $\hat{\mathcal{Q}} = \mathcal{Q}$.
		\EndFor
	\end{algorithmic}
\end{algorithm}
Next, random samples of transitions from the replay memory is fed into the deep dueling neural network for training, and then the Q-value function is obtained by using (\ref{output_average}). Note that the features of the input layer are the states of the queue, the available edge nodes, and the size of the considered learning task, defined in (\ref{eq:state}). As such, all aspects of each system state are trained to improve the performance of the algorithm. Nevertheless, as proved in~\cite{mnih2015human}, during the training process, the Q-values for each pair of state and action will be changed. Thus, the algorithm may not be stable if a constantly shifting set of values is used to update the Q-network. To solve this problem, we use the \emph{quasi-static target network} method to improve the stability of the algorithm. In particular, we implement a target Q-network and update its network parameters $(\bm{\alpha^-}, \bm{\beta^-})$ with the Q-network parameters $(\bm{\alpha}, \bm{\beta})$ after every $C$ steps. The target network parameters remain unchanged between individual updates. For each transition $j$ in the random samples, we denote $y_j =  r_j + \gamma\max_{a_{j+1}}\mathcal{Q}(s_{j+1},a_{j+1};\bm{\alpha^-}, \bm{\beta^-})$ as the target value in the training process. Then, the loss function can be expressed as follows:
\begin{equation}
	\label{lossfunction}
	L_j(\bm{\alpha}, \bm{\beta})\!\!=\!\!\mathbb{E}_{(s_j,a_j,r_j,s_{j+1})\sim U(\mathbf{D})}\big[ \big( y_j -\mathcal{Q}(s_j,a_j;\bm{\alpha}, \bm{\beta})\big)^2\big],
\end{equation}
where $\gamma$ is the discount factor. Differentiating the loss function in~(\ref{lossfunction}) with respect to the neural network's parameters, we have the following gradient:
\begin{equation}
\label{gradient_loss}
\begin{aligned}
\nabla_{\bm{\alpha}, \bm{\beta}}L(\bm{\alpha}, \bm{\beta})=\mathbb{E}_{(s_j,a_j,r_j,s_{j+1})} \big[\big(y_j -\mathcal{Q}(s_j,a_j;\bm{\alpha}, \bm{\beta})\nabla_{\alpha,\beta}\mathcal{Q}(s_j,a_j;\bm{\alpha}, \bm{\beta})\!\big)\big].
\end{aligned}
\end{equation}
By using~(\ref{gradient_loss}), the loss function is minimized to update the parameters of the deep dueling network. After a number of iterations, the algorithm can obtain the optimal coding and scheduling policy for the system. Algorithm~\ref{deepduelingqlearning} provides details of the deep dueling algorithm.
\subsubsection{Complexity Analysis}
In this work, the deep dueling neural network consists of input layer $L_0$, hidden layer $L_1$, and two streams to estimate the value and the advantage function. The value stream consists of layer $L_\mathrm{value}$ which is used to estimate the value function. The advantage stream consists of layer $L_\mathrm{advantage}$ which is used to estimate the advantage function. Let $|L_\mathrm{i}|$ denote the size (i.e., the number of neurons) of layer $L_\mathrm{i}$. We then can formulate the complexity of the deep dueling neural network as $|L_0||L_1| + |L_1||L_\mathrm{value}| + |L_1||L_\mathrm{advantage}|$. At each training step, a number of training samples, i.e., transitions, are randomly taken from the memory pool and fed to the deep dueling neural network for training. Thus, the total complexity of the training process is $\mathcal{O}\Big(IN_\mathrm{b}\Big(|L_0||L_1| + |L_1||L_\mathrm{value}| + |L_1||L_\mathrm{advantage}|\Big)\Big)$, where $N_\mathrm{b}$ is the size of the training batch and $I$ is the total number of training iterations. In our paper, the size of $L_0$ is the number of state features including the number of learning task currently waiting in the queue, the current task size, and the states of edge nodes, therefore $|L_0|= 2+N$. Hidden layer $L_1$ has 16 neurons. $|L_\mathrm{value}| = 1$ as this layer is used to estimate the value of the current state only. The size of $L_\mathrm{advantage}$ is the number of actions that the MEC can choose as the advantage layer is used to estimate the advantage function of all feasible actions in the current state. Clearly, the architecture of our deep dueling neural network is simple. Thus, it can be easily deployed at the MEC server which is usually equipped with sufficient computing resources. In the simulations, we show that with only 16 neurons in the hidden layers and the batch size is 16, our proposed deep dueling can converge to the optimal coding and scheduling policy much faster than those of the conventional Q-learning and deep Q-learning algorithms.
\section{Performance Analysis and Simulation Results}
\label{sec:evaluation}
\subsection{Parameter Setting}
In this work, we consider that the task queue at the MEC server can store up to 10 learning tasks. There are five edge nodes in the system to serve learning tasks. Unless otherwise stated, the task arrival probability is set at 0.7. The time for serving one data point is set at 5 milliseconds for all edge nodes~\cite{zhang2019model}. The size (i.e., number of data points) of each learning task is randomly taken from the set of $\{100, 200, 300\}$. We set $\mathbf{p} = \{0.1, 0.5, 0.2, 0.3, 0.9\}$ and $\bm{\lambda}=\{0.1, 1, 0.5, 0.2, 2\}$. All the above parameters are then varied in the next section to evaluate the performance of our proposed algorithm in various scenarios. It is worth noting that our proposed deep dueling algorithm does not require to know these parameters in advance. Instead, it can interact with the environment, observe results, and then learn these parameters to obtain the optimal coding and scheduling policy for the system in a real-time manner.

The architecture of the deep neural network plays an important role in the learning process, and thus it needs to be carefully designed~\cite{goodfellow2016deep}. In particular, with more hidden layers, the algorithm can learn the problem better. However, the complexity of the algorithm will increase, resulting in a long training period. Moreover, a high number of hidden layers does not always guarantee a good learning performance due to the overestimation problems of optimizers. Similarly, the number of neutrons in each layer as well as the mini-batch size are also required a thoughtful design. For the deep Q-learning, we deploy two fully-connected hidden layers connected to the input layer and the output layer. For the deep dueling algorithm, the neural network consists of two streams to separately estimate the value function and the advantage function. These two streams are connected to a shared hidden layer (after the input layer) as illustrated in Fig.~\ref{Fig.deepduelingqlearning}. The sizes of all hidden layers are set at 16. The mini-batch size is set at 16. The maximum size of the experience replay buffer is $10,000$. The target Q-network is updated after every $1,000$ learning steps. We use the $\epsilon$-greedy scheme with the initial value of $\epsilon$ is 1, and its final value is $0.01$. The decay factor is set at $0.9999$. The learning rate and the discount factor of the deep dueling and the deep Q-learning algorithms are set at $0.0001$ and $0.99$, respectively. For the Q-learning, the learning rate and the discount factor are set at $0.1$ and $0.9$, respectively.

To evaluate the proposed solution, we compare its performance with three other approaches: (i) \textit{Greedy}, (ii) \textit{OneNode}, and (iii) \textit{Static Optimal Code}.

\begin{itemize}
	
	\item \textit{Greedy:} For this policy, the MEC server selects all available edge nodes to serve a learning task. The task is then coded and equally distributed to all the available edge nodes. As the MEC server requires results from all assigned edge nodes to successfully decode the final result. This policy is to evaluate the impact of straggling edge nodes and wireless links.
	
	\item \textit{OneNode:} In this policy, the MEC server randomly selects one available edge node to serve a learning task. This policy is used to evaluate the typical uncoded and non-distributed learning approaches.
	
	\item \textit{Static Optimal Code}: This policy is based on the optimal MDS code proposed in~\cite{lee2017speeding}. In particular, given $\mathbf{E}_\mathrm{av}$ edge nodes, the optimal value of $k$ is derived as follows:
	\begin{equation}
		\label{eq:optimalMDS}
		k^\dagger = \Bigg[ 1 + \frac{1}{W_{-1}(-e^{-\hat{\lambda}-1})} \Bigg] |\mathbf{E}_\mathrm{av}|,
	\end{equation}
where $W_{-1}(.)$ is the lower branch of Lambert $W$ function and $\hat{\lambda}$ presents the average straggling parameter of all edge nodes. By using this equation, the MEC server can select $(n=|\mathbf{E}_\mathrm{av}|, k=k^\dagger)$ MDS code for each learning task. This policy is used to show the performance of a static optimal code that does not consider the heterogeneous structure of both edge nodes and wireless links, e.g., in channel/backhaul link quality to/from the MEC server, straggling effects and computing capabilities of edge nodes.
\end{itemize}

We also obtain the policy of the proposed solution by running the conventional Q-learning algorithm~\cite{watkins1992q} to compare the effectiveness of the proposed deep dueling algorithm. The performance metrics for evaluating the proposed approach are the average number of learning tasks in the queue per time slot, the average task dropping probability, and the average delay in the system of a learning task. The average task dropping probability corresponds to the average number of dropped learning tasks in each time slot when the task queue is full. The average delay of a learning task in the system is calculated from the time a learning task arrives at the system until the task leaves the system (i.e., the MEC server finishes serving the task).
\subsection{Simulation Results}
\subsubsection{Convergence of Learning Algorithms}
\begin{figure}[!]
	\centering
	\includegraphics[scale=0.4]{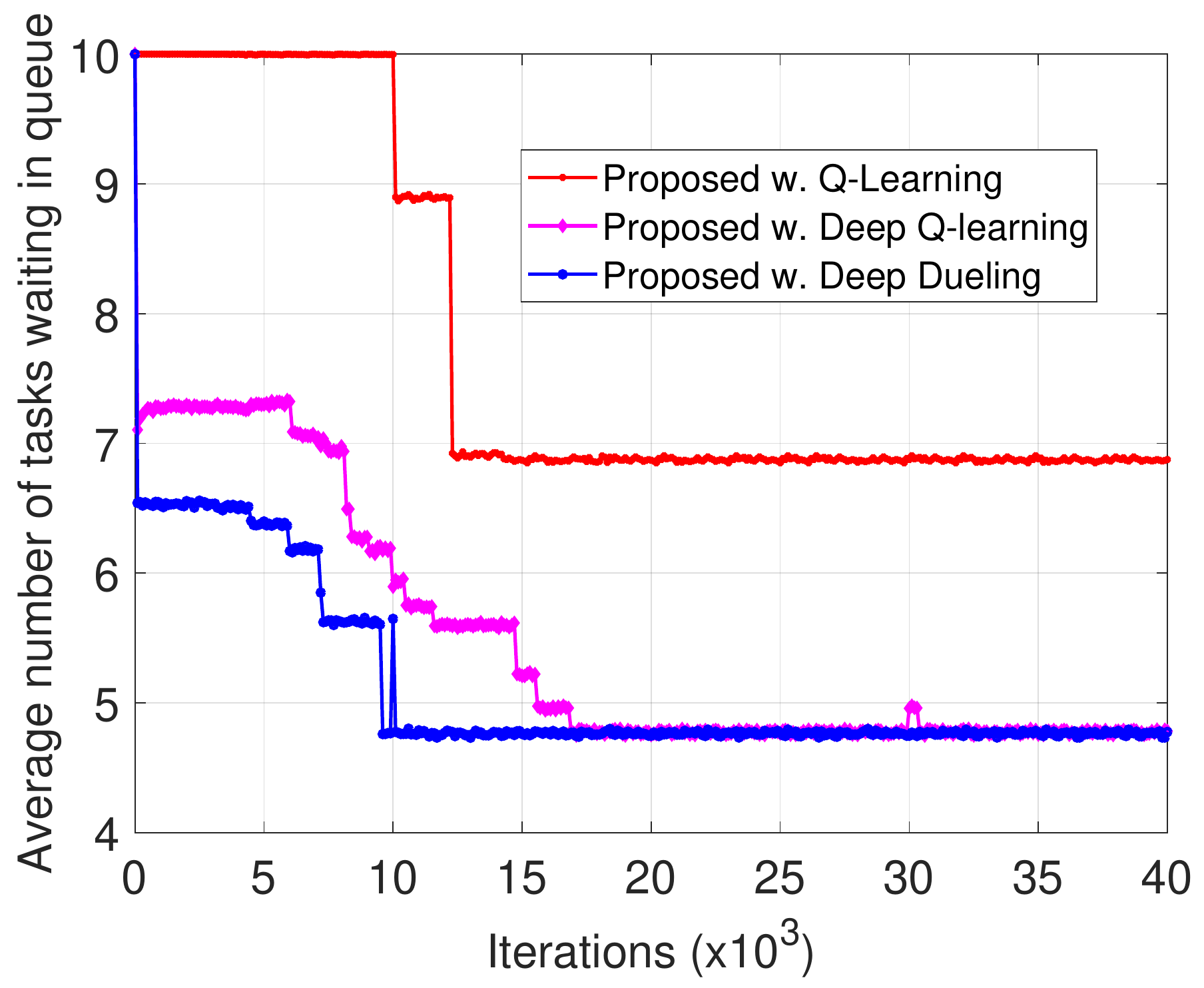}
	\caption{Convergence rates of learning algorithms.}
	\label{Fig.Convergence}
\end{figure}
In Fig.~\ref{Fig.Convergence}, we show the learning process of the Q-learning, deep Q-learning, and deep dueling algorithms. As can be observed, the convergence rate of the Q-learning algorithm is much slower than those of the deep Q-learning and deep dueling algorithms. This is stemmed from the fact that the Q-learning algorithm has a very slow-convergence due to the curse-of-dimensionality problem, especially in complicated systems as considered in our work. By using the novel deep dueling neural network architecture, our deep dueling algorithm achieves the best convergence rate. In particular, the deep dueling algorithm can converge to the optimal coding and scheduling policy within $10,000$ iterations, while the deep Q-learning algorithm needs more than $15,000$ iterations to converge to the optimal coding and scheduling policy. In the next section, all results of the deep dueling algorithm are obtained at $4 \times 10^4$ iterations, while those of the Q-learning algorithm are obtained at $10^6$ iterations. Note that the Q-learning algorithm is used as a benchmark to demonstrate the effectiveness of the proposed deep dueling algorithm.
\subsubsection{System Performance}
\begin{figure*}[!]
	\centering
	\begin{subfigure}[b]{0.3\textwidth}
		\centering
		\includegraphics[scale=0.27]{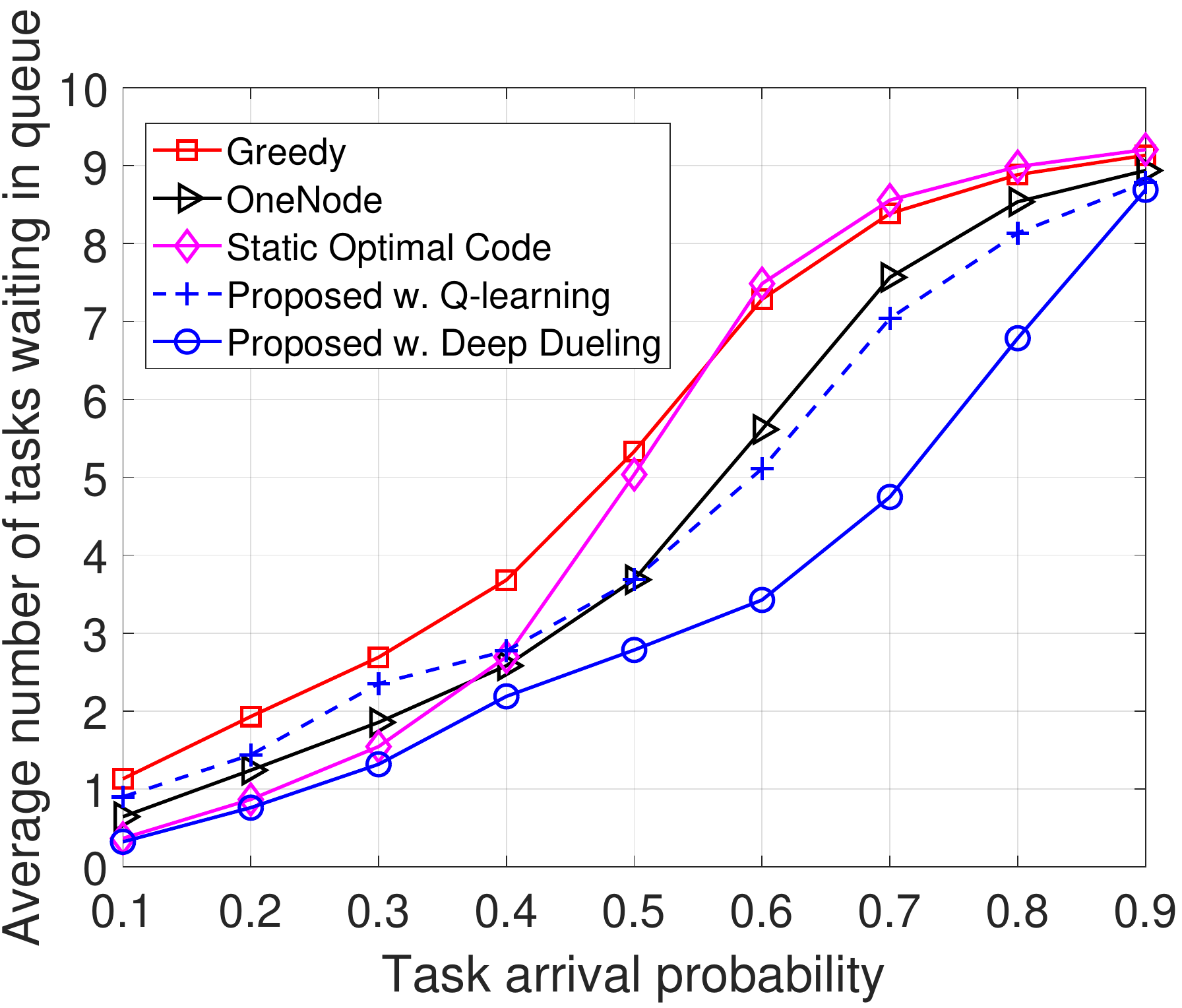}
		\caption{}
	\end{subfigure}%
	~
	\begin{subfigure}[b]{0.3\textwidth}
		\centering
		\includegraphics[scale=0.27]{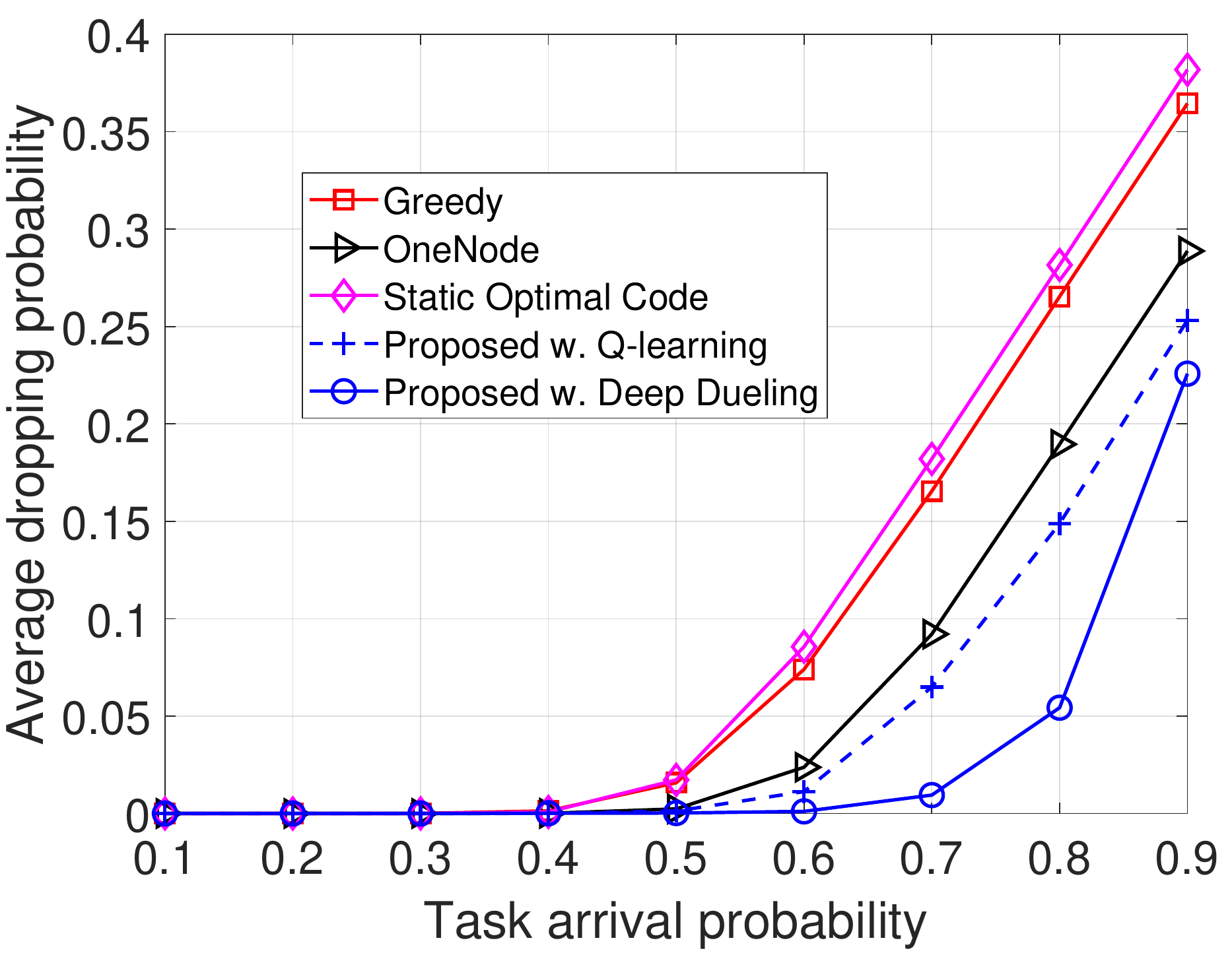}
		\caption{}
	\end{subfigure}%
	~
	\begin{subfigure}[b]{0.3\textwidth}
		\centering
		\includegraphics[scale=0.27]{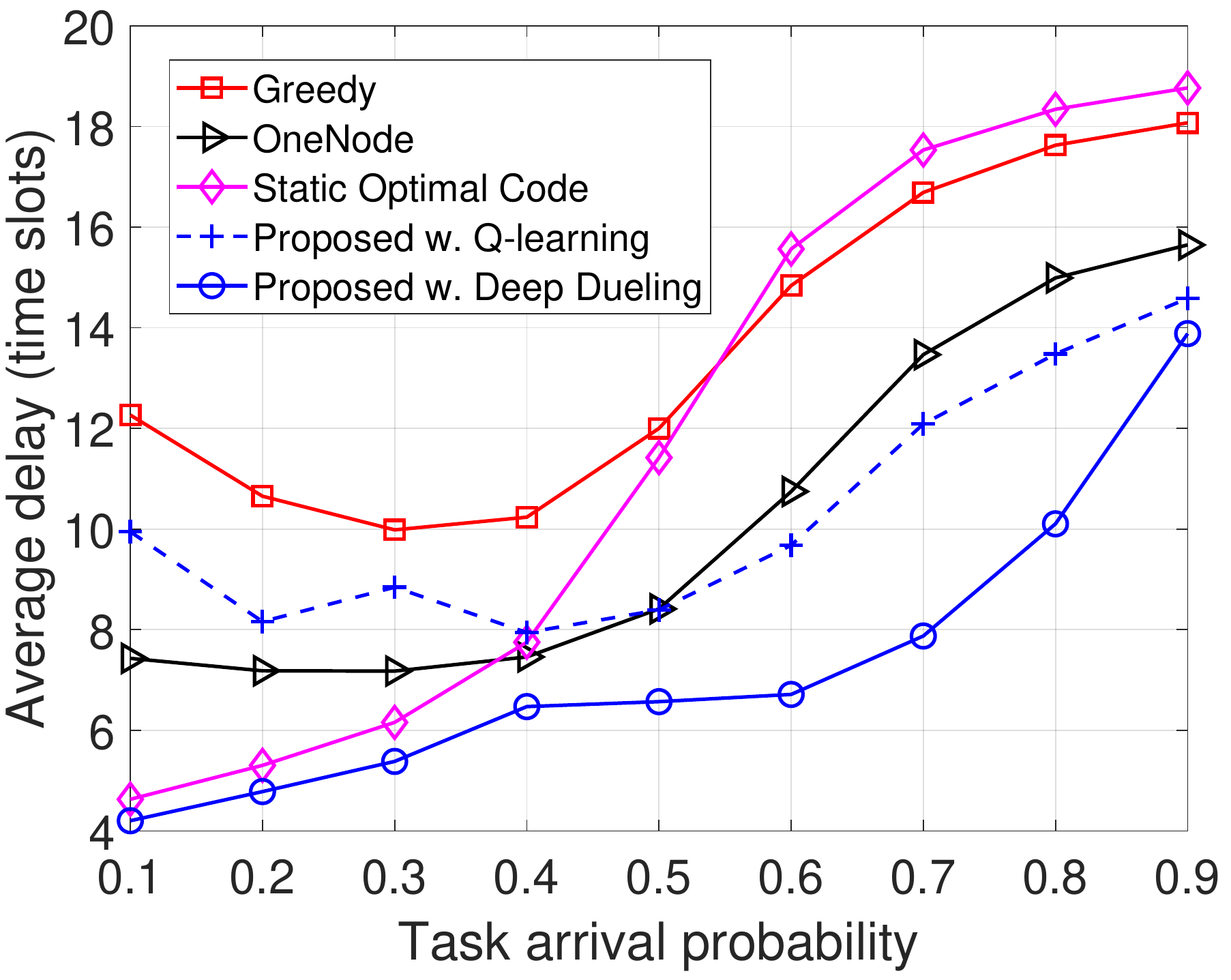}
		\caption{}
	\end{subfigure}%
	\caption{(a) Average number of tasks waiting in the queue, (b) task dropping probability, and (c) average delay of learning tasks in the system vs. task arrival probability.} 
	\label{fig:varyArrival}
\end{figure*}

In this section, we perform simulations to evaluate the performance of the proposed solution in terms of the number of tasks waiting in the queue, the average task dropping probability, and the average delay of learning tasks in the system in various scenarios. First, we vary the arrival probability of learning tasks and compare the performance of the proposed solution with those of \textit{Greedy}, \textit{OneNode}, and \textit{Static Optimal Code} policies as shown in Fig.~\ref{fig:varyArrival}. Clearly, when the task arrival probability increases, the average number of learning tasks in the queue increases as there are more learning tasks arriving at the system as shown in Fig.~\ref{fig:varyArrival}(a). With the proposed deep dueling algorithm, our solution can reduce the number of tasks in the queue by up to 71\%, 50\%, and 54\% compared to the \textit{Greedy}, \textit{OneNode}, and \textit{Static Optimal Code}, respectively. The reason is that the deep dueling algorithm can learn and maximize the number of learning tasks served by the edge nodes by determining the optimal MDS code as well as the best edge nodes with stable wireless links. Fig.~\ref{fig:varyArrival}(b) also confirms the outperformance of our proposed solution in terms of task dropping probability. Interestingly, in Fig.~\ref{fig:varyArrival}(c), when the task arrival probability increases from $0.1$ to $0.3$, the average delay of learning tasks in the system under the \textit{Greedy} and \textit{OneNode} schemes decrease by nearly 20\% and 3\%, respectively. The reason is that under these policies, the MEC server randomly chooses edge nodes to serve learning tasks. As such, there are cases in which learning tasks are severed by highly-straggling edge nodes and/or unstable wireless links. However, when there are more learning tasks arriving at the system, these learning tasks are likely severed by good edge nodes and stable links as the straggling devices may take longer time to serve other learning tasks, and thus they may not be often available. As a result, the average waiting time in the system of a learning task reduces when the arrival probability increases from $0.1$ to $0.3$. The \textit{Static Optimal Code} does not encounter this trend and achieves a better performance compared to the \textit{Greedy} and \textit{OneNode} policies, thanks to the use of MDS code. However, when the task arrival probability is higher than $0.3$, the performance of the \textit{Static Optimal Code} is similar to that of the \textit{Greedy} policy and lower than that of the \textit{OneNode} policy. The reason is that when there are many learning tasks arriving at the system, the computing resources of edge nodes are mostly utilized. Thus, at each time slot, the available edge nodes are likely the nodes with good computing power (as they can finish their assigned task quickly and become available for new tasks). Thus, sending a task to a single edge node with high computing power for processing is better than distributing it to several edge nodes with different wireless connections. Note that the \textit{Static Optimal Code} policy does not account for the effects of unstable wireless links. Nevertheless, by learning and avoiding choosing the slow edge nodes and unstable wireless links, our proposed solution always achieves the best performance. In particular, the average delay of learning tasks obtained by our solution is much lower than those of the \textit{Greedy}, \textit{OneNode}, and \textit{Static Optimal Code} policies, reduced by up to 66\% as shown in Fig.~\ref{fig:varyArrival}(c).

\begin{figure*}[!]
	\centering
	\begin{subfigure}[b]{0.3\textwidth}
		\centering
		\includegraphics[scale=0.27]{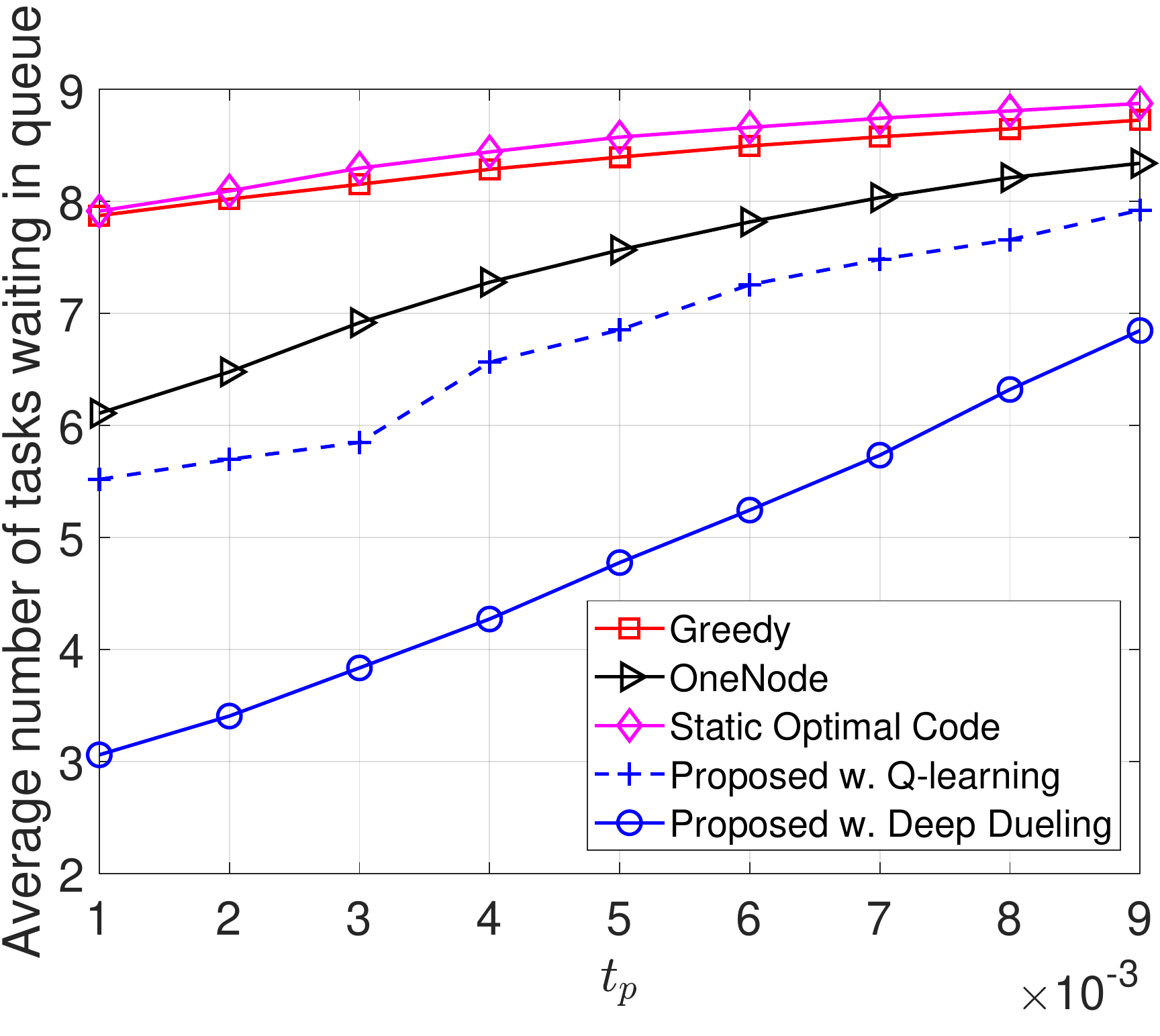}
		\caption{}
	\end{subfigure}%
	~
	\begin{subfigure}[b]{0.3\textwidth}
		\centering
		\includegraphics[scale=0.27]{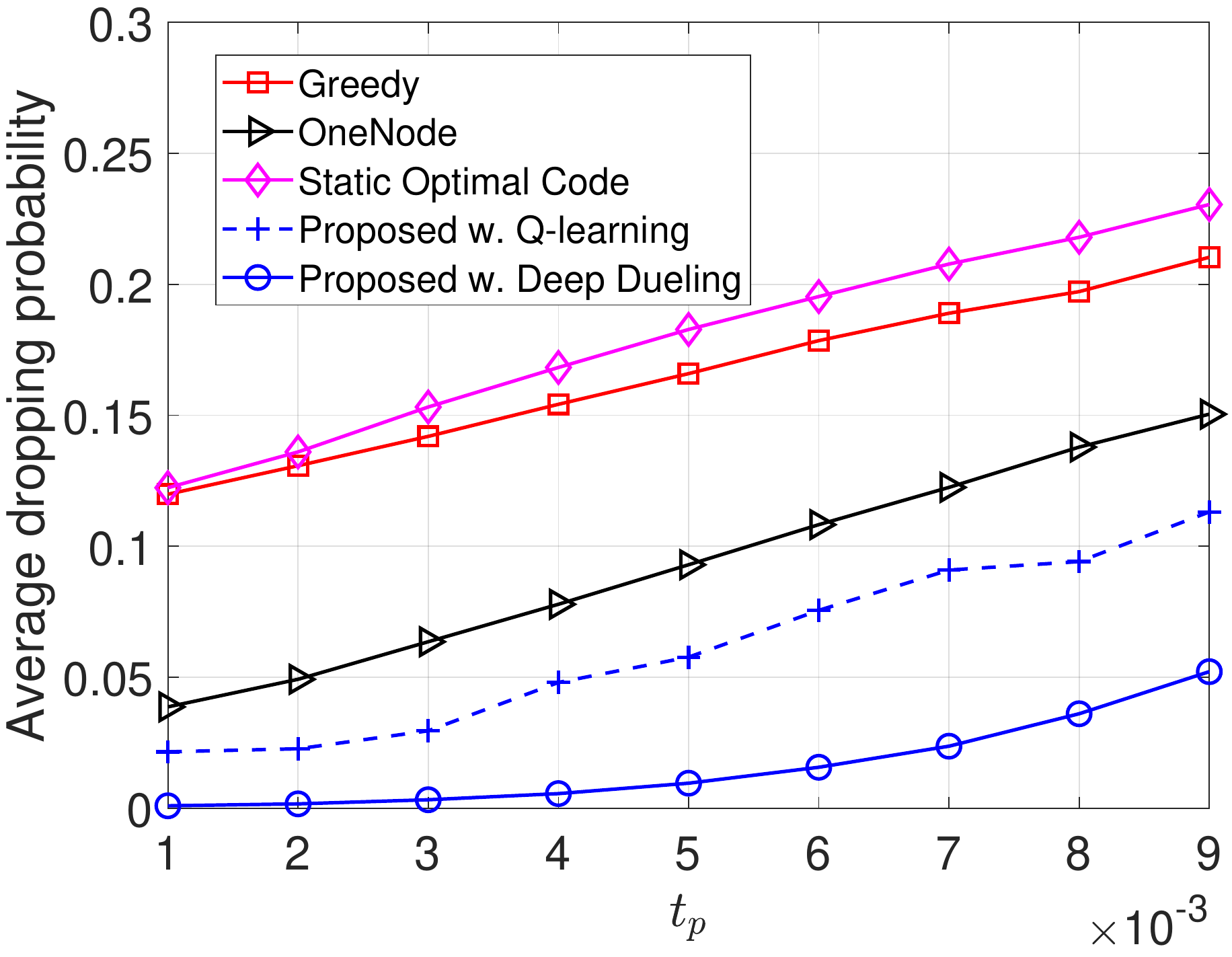}
		\caption{}
	\end{subfigure}%
	~
	\begin{subfigure}[b]{0.3\textwidth}
		\centering
		\includegraphics[scale=0.27]{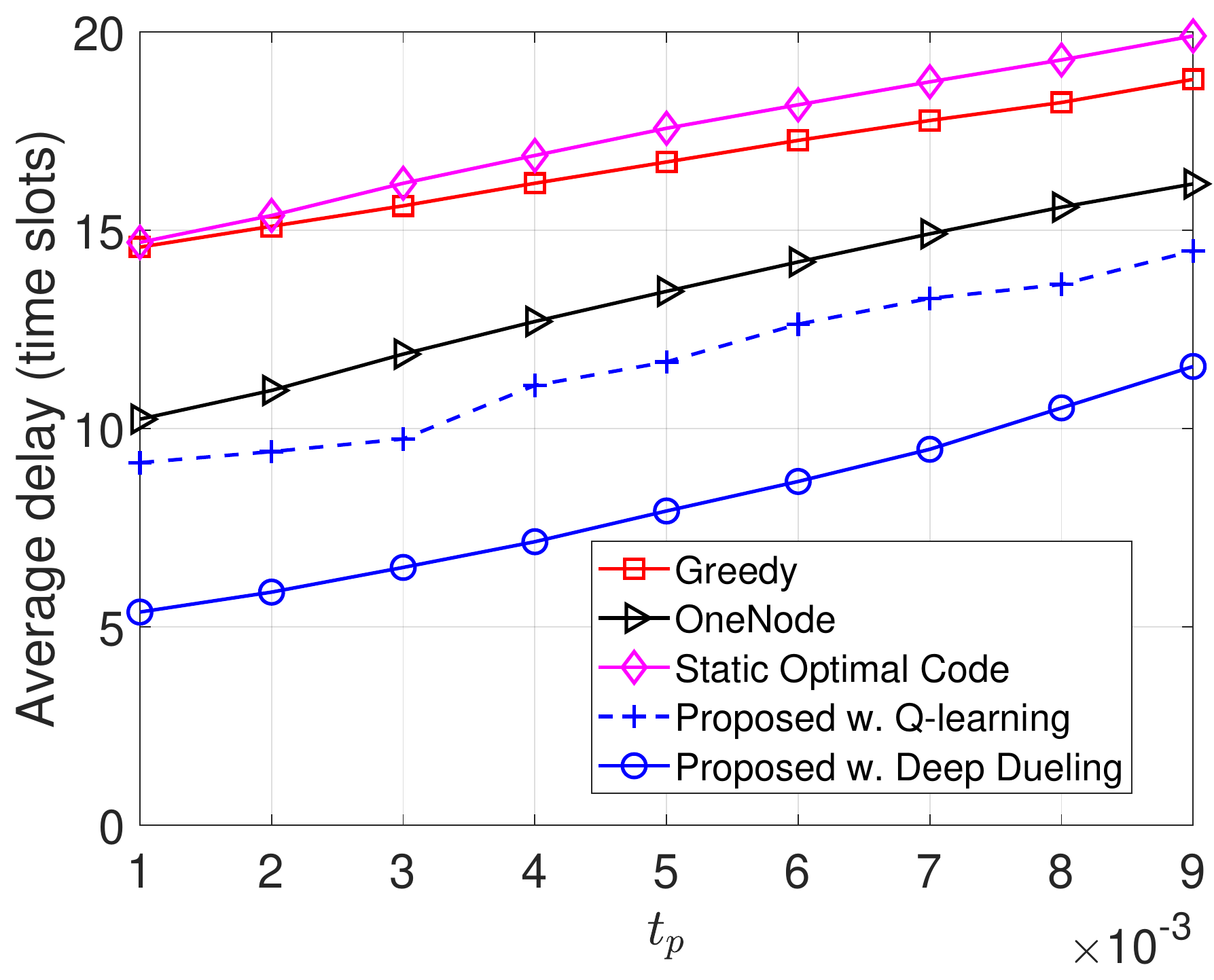}
		\caption{}
	\end{subfigure}%
	\caption{(a) Average number of tasks waiting in the queue, (b) task dropping probability, and (c) average delay of learning tasks in the system vs. processing time.}
	\label{fig:varyComputing}
\end{figure*}

Next, we vary the time for serving one data point (e.g., one matrix row) $t_p$ and evaluate the system performance as shown in Fig.~\ref{fig:varyComputing}. It can be observed that when the processing time increases, the system performances obtained by all methods significantly decrease. The reason is that, with higher processing time, the edge nodes need more time to serve learning tasks, and thus increasing the serving time of learning tasks. Consequently, learning tasks need to wait longer in the queue. Nevertheless, in all scenarios, the proposed solution always achieves the best performance and can reduce the average delay of learning tasks by 63\%, 47\%, and 63\% compared to the \textit{Greedy}, \textit{OneNode}, and \textit{Static Optimal Code}, as shown in Fig.~\ref{fig:varyComputing}(c) respectively. The reason is that the deep dueling algorithm can learn and adapt with the environment parameters in order to select the optimal MDS code for each learning task as well as avoid unstable wireless links. Among all policies, the \textit{Static Optimal Code} possesses the worst performance as this policy obtains the optimal MDS code (see (\ref{eq:optimalMDS})) without considering the heterogeneity of edge nodes and wireless links. This also confirms out analyses on the drawback of existing static coding mechanisms, i.e., they are only effective under specific scenarios and assumptions. Our proposed solution, otherwise, can learn all the unpredictable parameters of the edge nodes and wireless links to jointly optimize coding and scheduling policies for learning tasks. It is worth noting that the Q-learning algorithm cannot obtain the optimal coding and scheduling policy at $10^6$ iterations as discussed above. Thus, the results obtained by the Q-learning algorithm are not as good as those of the proposed deep dueling algorithm.

\begin{figure*}[!]
	\centering
	\begin{subfigure}[b]{0.3\textwidth}
		\centering
		\includegraphics[scale=0.27]{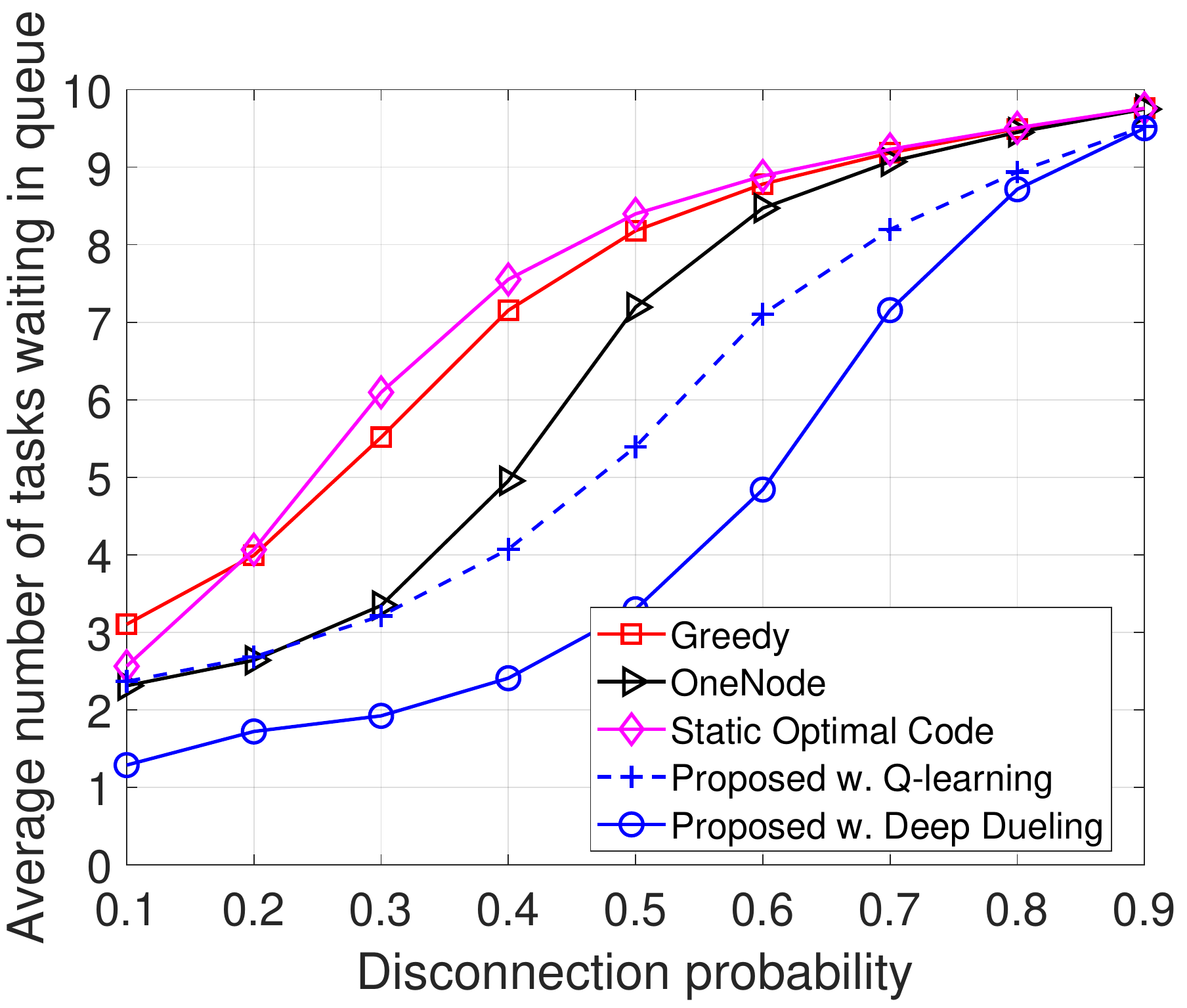}
		\caption{}
	\end{subfigure}%
	~
	\begin{subfigure}[b]{0.3\textwidth}
		\centering
		\includegraphics[scale=0.27]{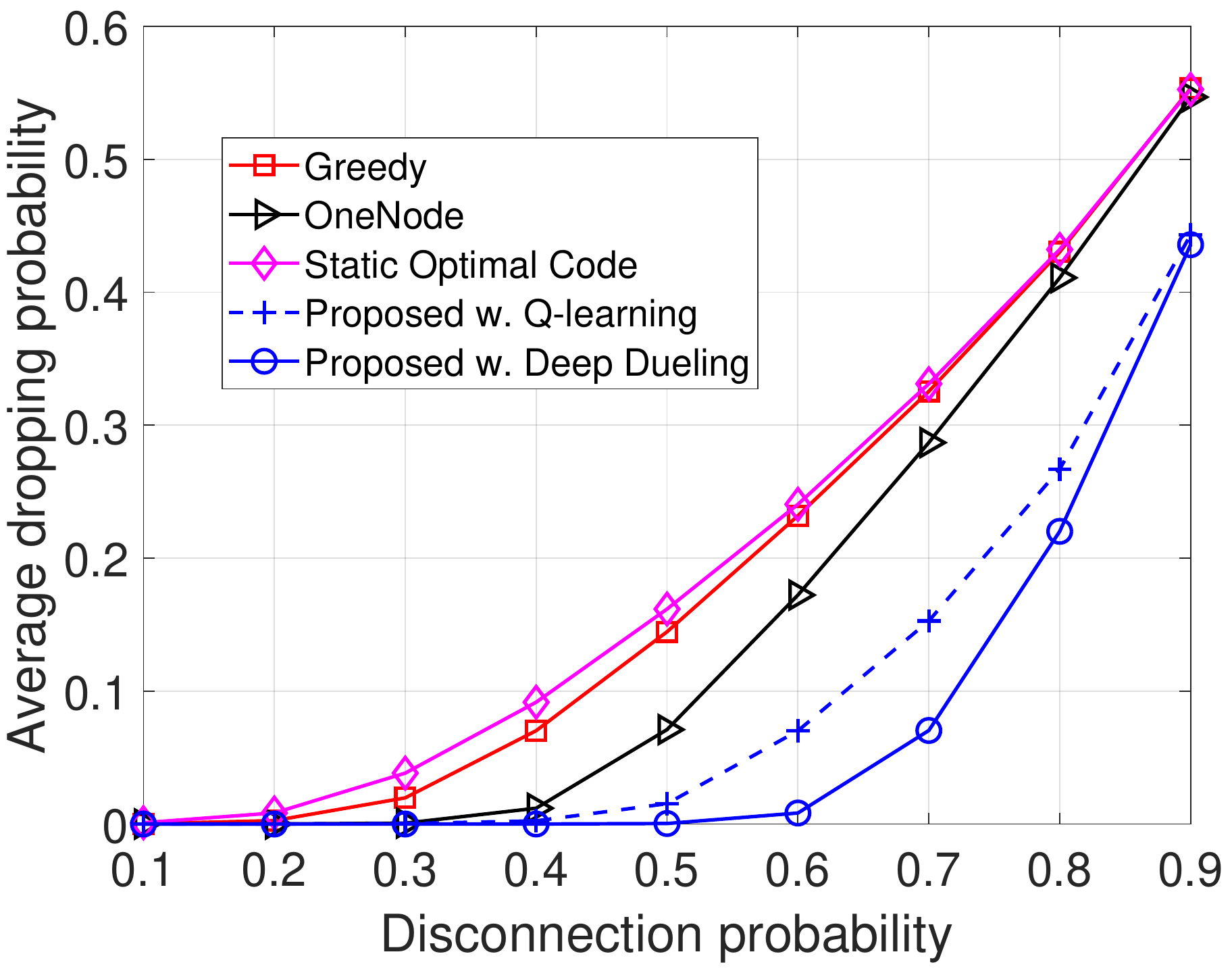}
		\caption{}
	\end{subfigure}%
	~
	\begin{subfigure}[b]{0.3\textwidth}
		\centering
		\includegraphics[scale=0.27]{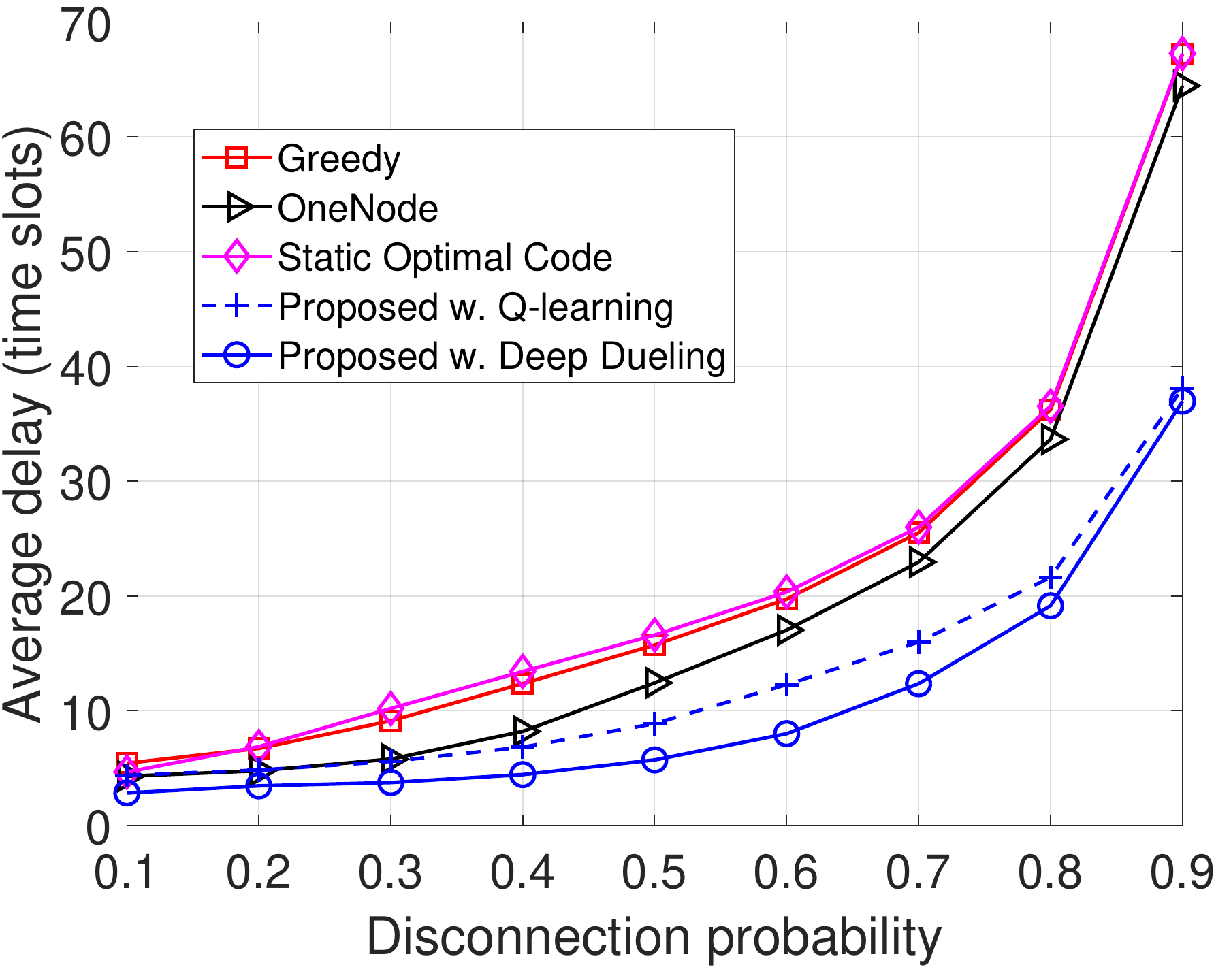}
		\caption{}
	\end{subfigure}%
	\caption{(a) Average number of tasks waiting in the queue, (b) task dropping probability, and (c) average delay of learning tasks in the system vs. disconnection probability of links.}
	\label{fig:varyDisconnect}
\end{figure*}

In Fig.~\ref{fig:varyDisconnect}, we vary the disconnection probability of wireless links and observe the system performance in terms of the number of tasks waiting in the queue, task dropping probability, and the average delay of learning tasks in the system. Clearly, when the disconnection probability increases, system performances obtained by all the policies drop dramatically. This is stemmed from the fact that when the wireless links are more unstable, the MEC server and the edge results need more time to resend the sub-learning task and results, respectively. Consequently, the serving time of learning tasks increases, resulting in a higher delay for learning tasks. It is worth noting that when the disconnection probability increases from $0.1$ to $0.6$, the performance obtained by the \textit{OneNode} policy is much better than that of the \textit{Greedy} policy. The reason is that under the \textit{Greedy} policy, the MEC server and edge nodes need to resend data when the wireless links are disconnected. Differently, under the \textit{OneNode} policy, each learning task is served by only one edge node, and thus the frequency of resending data is lower than that of the \textit{Greedy} policy, resulting in a better performance compared with that of the \textit{Greedy} policy. However, when the disconnection probability is high, i.e., higher $0.7$, the performance gap between these two policies is very small as all links are likely disconnected. Nevertheless, our proposed solution always achieves the best performance in all cases compared to those of the \textit{Greedy} and \textit{OneNode} policies. The reason is that the proposed deep dueling can learn from the environment and avoid choosing highly-straggling edge nodes as well as adapt its optimal coding and scheduling policy when the disconnection probability changes. Again, the \textit{Static Optimal Code} achieves the worst performance as this policy does not account for the effects of unstable wireless links when obtaining the optimal MDS code for each learning task as expressed in (\ref{eq:optimalMDS}).

\begin{figure*}[!]
	\centering
	\begin{subfigure}[b]{0.3\textwidth}
		\centering
		\includegraphics[scale=0.27]{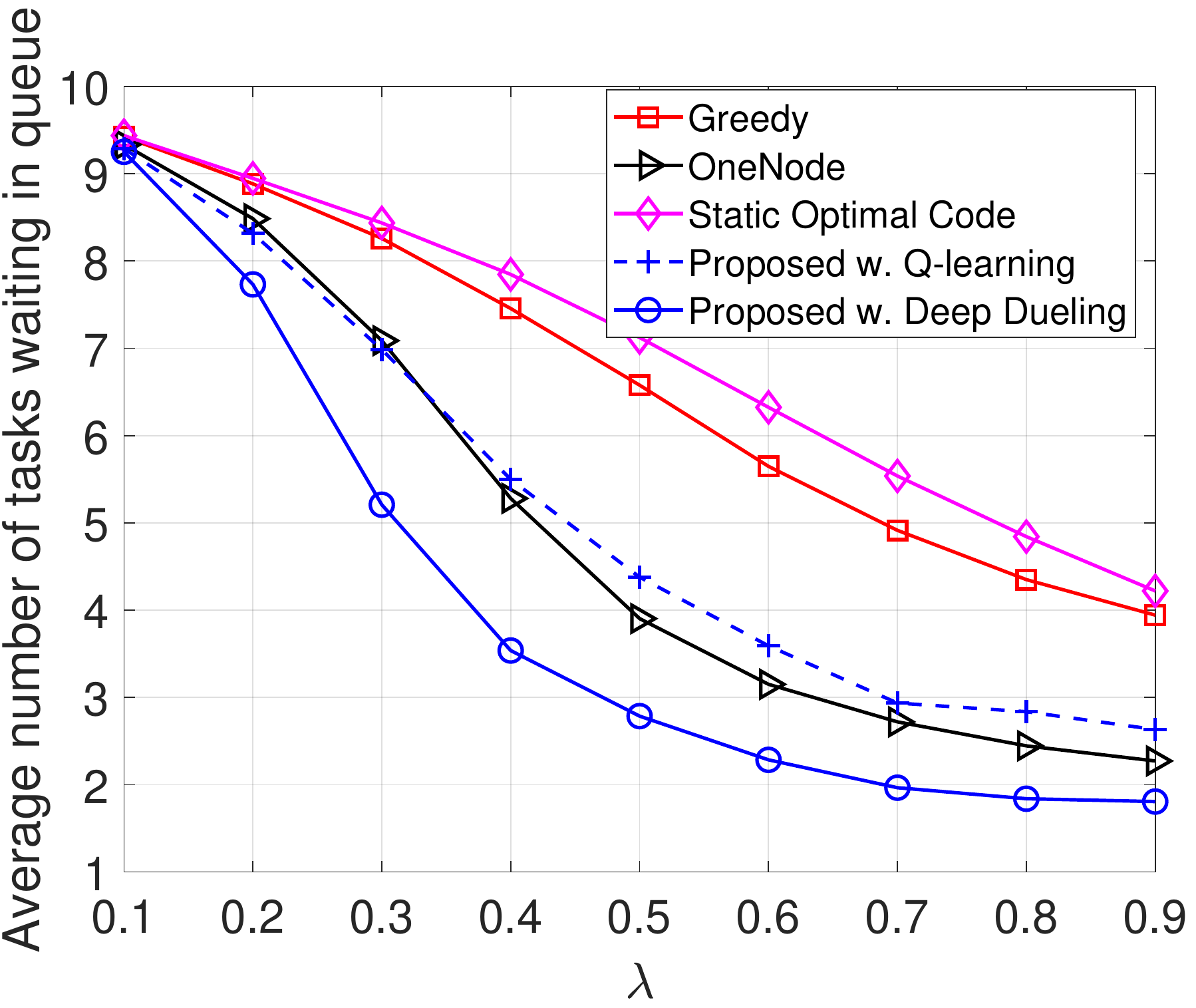}
		\caption{}
	\end{subfigure}%
	~
	\begin{subfigure}[b]{0.3\textwidth}
		\centering
		\includegraphics[scale=0.27]{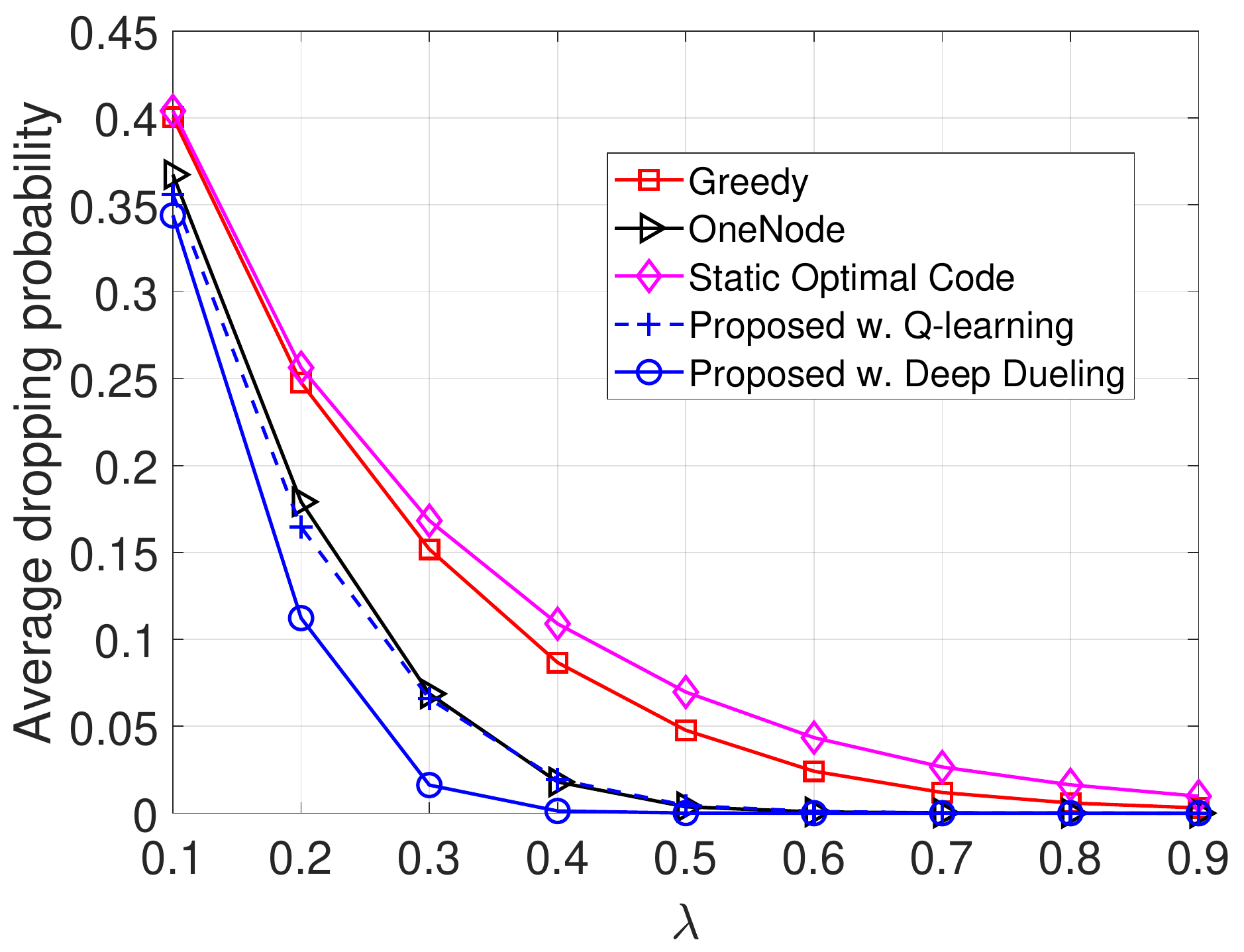}
		\caption{}
	\end{subfigure}%
	~
	\begin{subfigure}[b]{0.3\textwidth}
		\centering
		\includegraphics[scale=0.27]{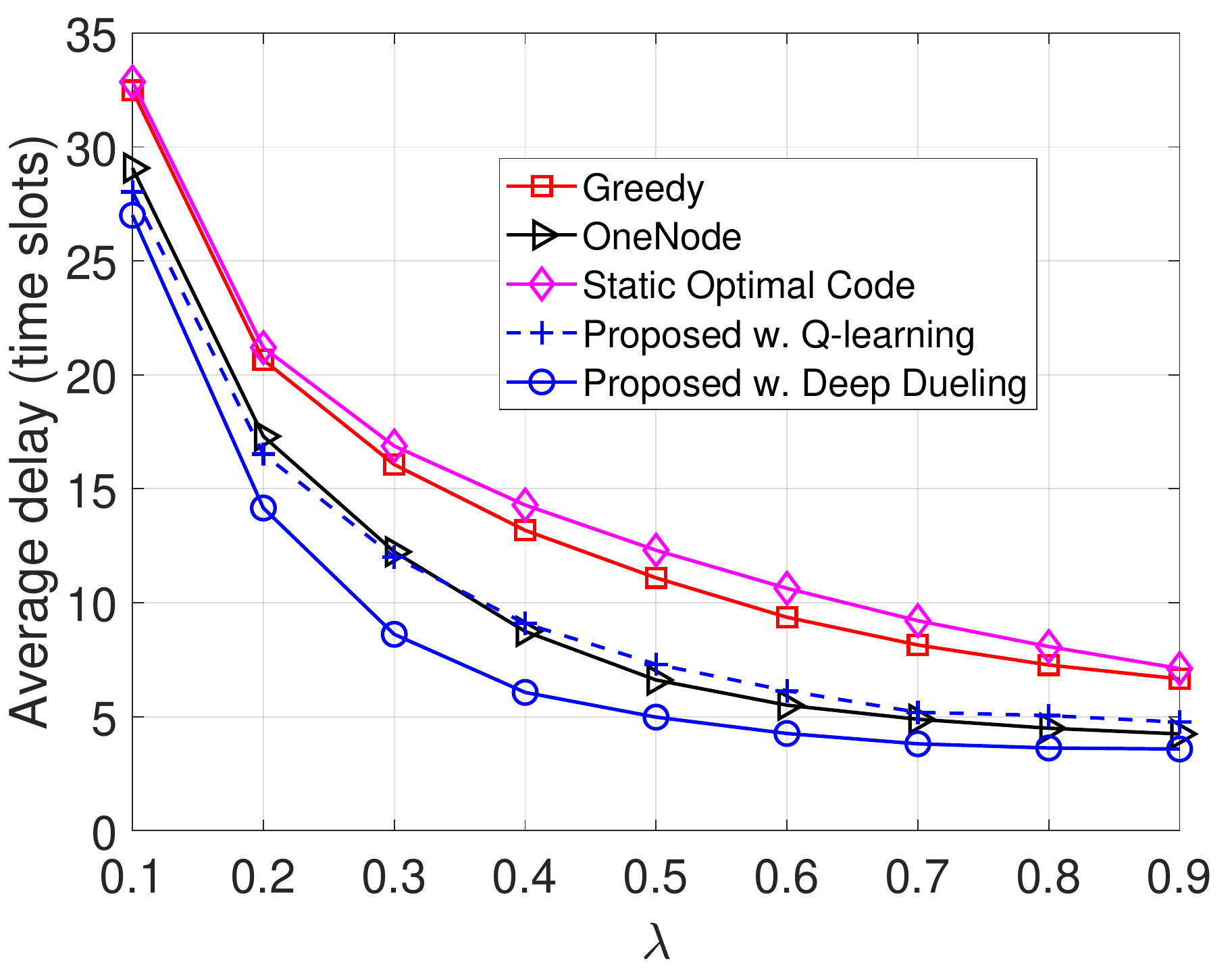}
		\caption{}
	\end{subfigure}%
	\caption{(a) Average number of tasks waiting in the queue, (b) task dropping probability, and (c) average delay of learning tasks in the system vs. the rate parameter $\lambda$ (in the exponential distribution) of the stochastic computing time of edge nodes.}
	\label{fig:varyStraggling}
\end{figure*}

Next, we vary the rate parameter $\lambda$ (in the exponential distribution) of the stochastic computing time of edge nodes and observe the system performance in Fig.~\ref{fig:varyStraggling}. Recall that, a lower value of $\lambda$ result in a longer time for stochastic computing. As such, when $\lambda$ increases, system performances obtained by all the policy will be decreased. Moreover, when $\lambda$ is small, the gap between solutions is small. Nevertheless, when $\lambda$ is increased, the gap will be enlarged. The reason is that, with lower values of $\lambda$, the edge nodes require more time to execute learning tasks. As such, the resources of the system are likely to be fully utilized. In these cases, there are not many options for the MEC server to serve learning tasks, resulting in a small performance gap between solutions. However, in all the scenarios, our proposed solution can achieve the best performance as the optimal policy can avoid unstable wireless links and select the optimal MDS code for each learning tasks.

\begin{figure*}[!]
	\centering
	\begin{subfigure}[b]{0.3\textwidth}
		\centering
		\includegraphics[scale=0.27]{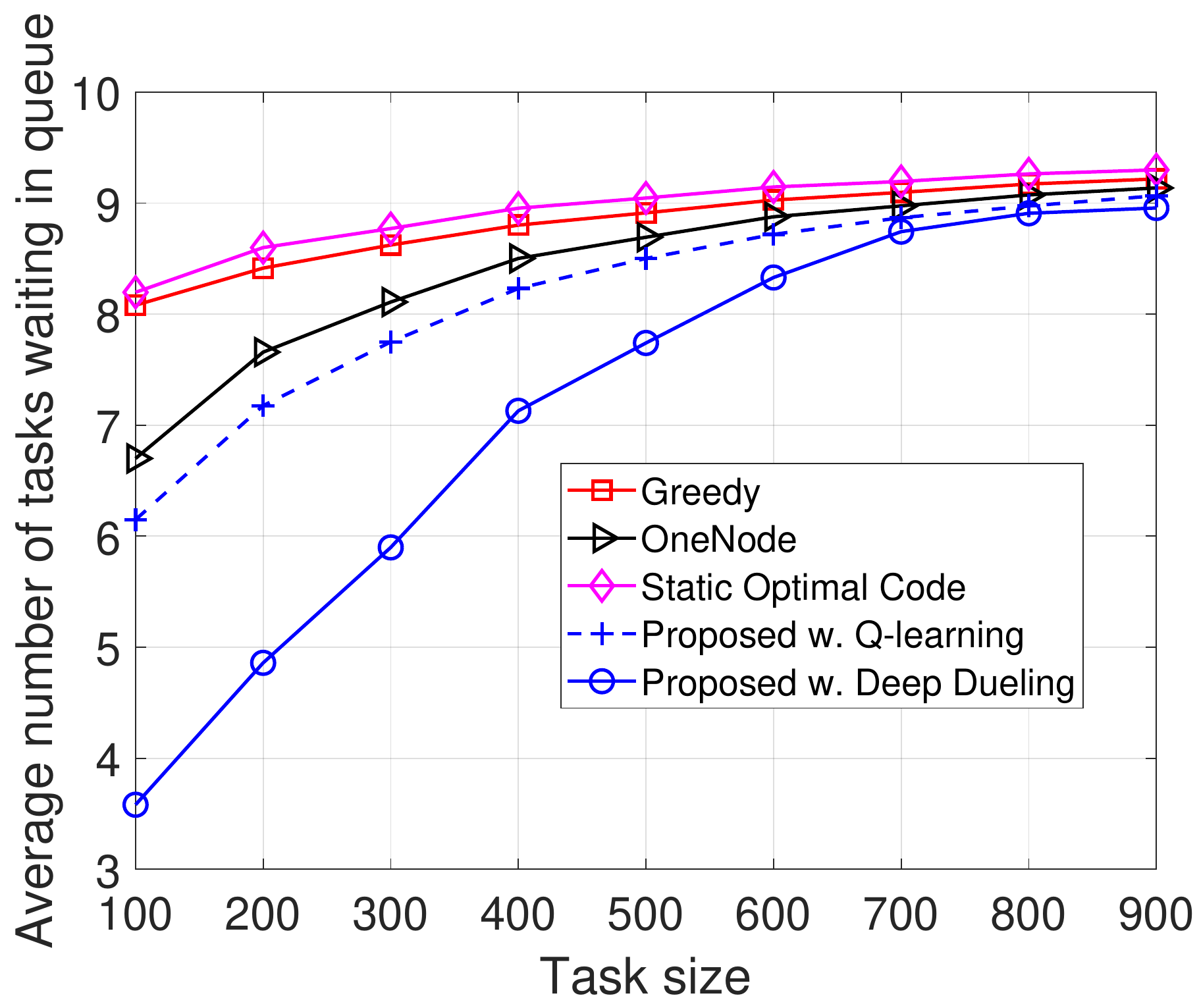}
		\caption{}
	\end{subfigure}%
	~
	\begin{subfigure}[b]{0.3\textwidth}
		\centering
		\includegraphics[scale=0.27]{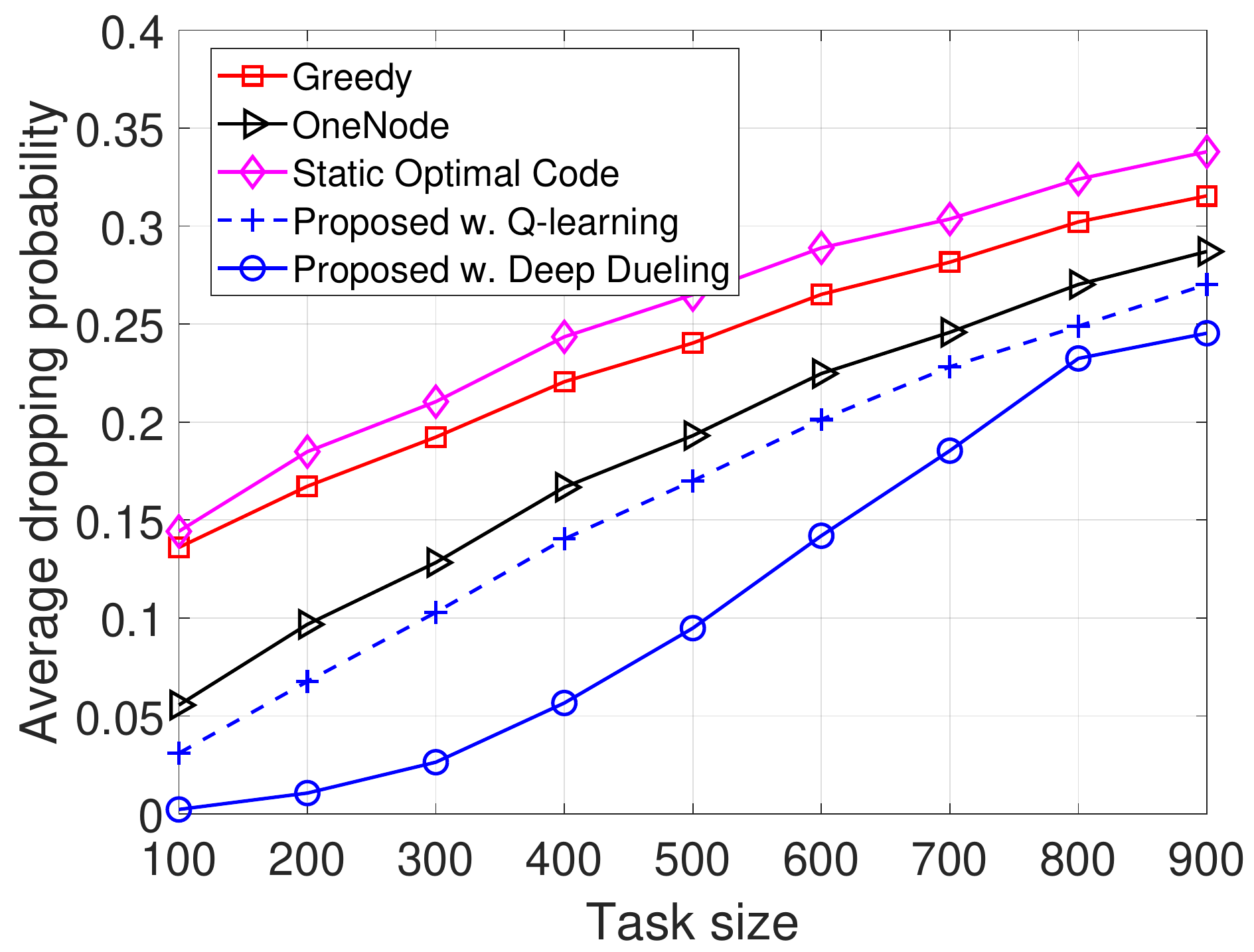}
		\caption{}
	\end{subfigure}%
	~
	\begin{subfigure}[b]{0.3\textwidth}
		\centering
		\includegraphics[scale=0.27]{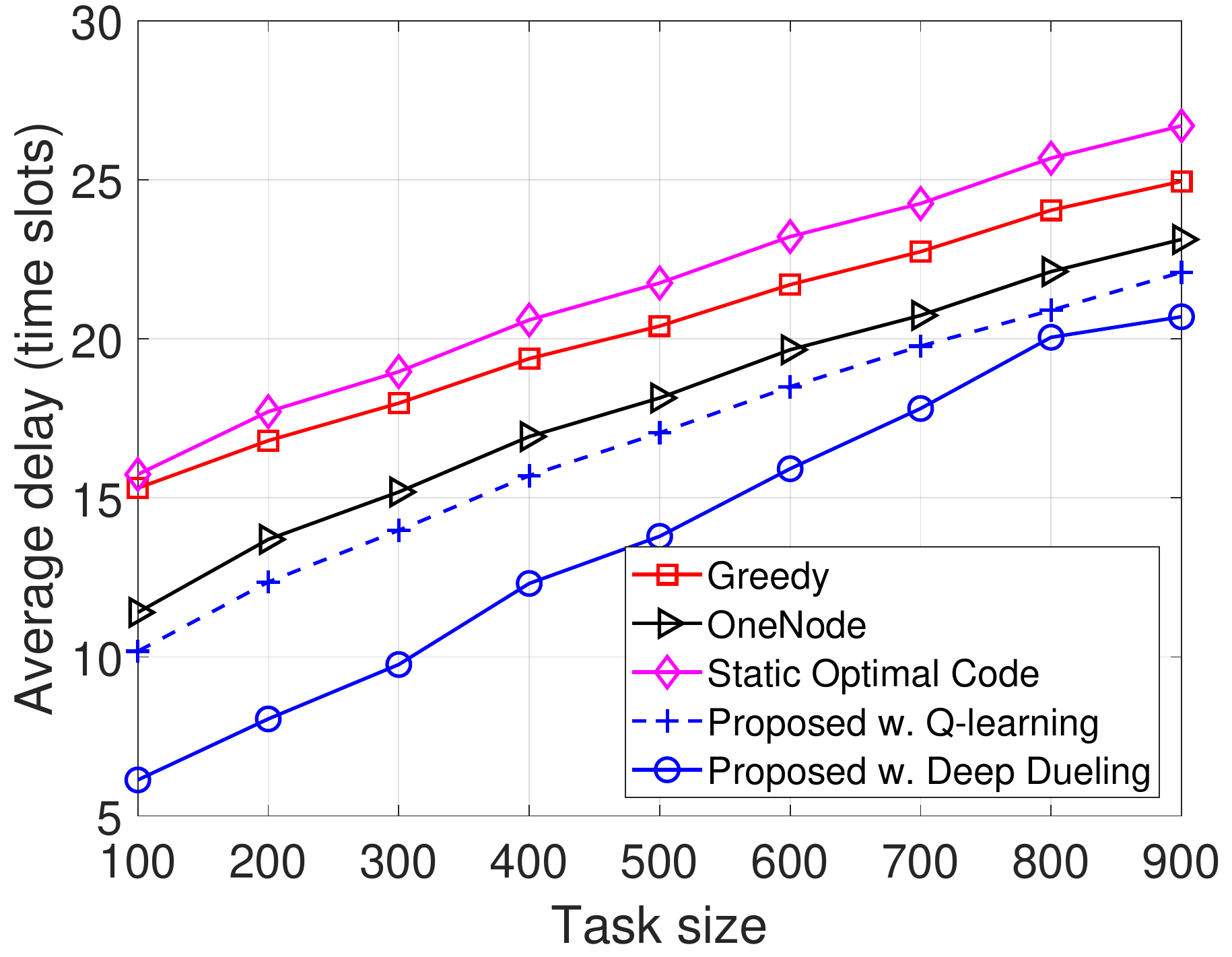}
		\caption{}
	\end{subfigure}%
	\caption{(a) Average number of tasks waiting in the queue, (b) task dropping probability, and (c) average delay of learning tasks in the system vs. task size.}
	\label{fig:varyTasksize}
\end{figure*}

Finally, we vary the data size of learning tasks and observe the system performances under different policies as shown in Fig.~\ref{fig:varyTasksize}. Clearly, when the task size increases, the system performances obtained by all the policies will be dropped. It is stemmed from the fact that with a larger task size, the edge nodes need more time to serve learning tasks. However, our proposed solution can always achieve the best performance compared to other policies because it can learn and select the best MDS code as well as the best edge nodes for each learning task. For example, when the task size is small, the deep dueling algorithm can select a small number of devices with more stable wireless links to serve a learning task. In contrast, when the task size is large, more edge nodes will be selected to reduce the average serving time of learning tasks. It can be observed that our proposed solution can reduce the average number of tasks waiting in the queue by 60\%, 46\%, and 61\% compare to those of the \textit{Greedy}, \textit{OneNode}, and \textit{Static Optimal Code} policies, respectively. Again, the \textit{Static Optimal Code} achieves the worst performance as this policy does not consider the learning task size and the effects of unstable wireless links.
\section{Conclusion}
\label{sec:conclusion}
In this paper, we have proposed a novel framework which can effectively address key challenges for the development of distributed learning in wireless edge networks. Specifically, we have first introduced a distributed learning model utilizing the recent advances in coded computing to mitigate the straggling problems on both the wireless links and the edge nodes. With the proposed distributed learning model, a learning task is first encoded into sub-learning tasks, and the sub-learning tasks are then transmitted to edge nodes for executing. This solution allows to significantly mitigate straggling problems caused by straggling edge nodes as well as unstable links between the MEC server and edge nodes. Furthermore, to deal with the dynamics and uncertainty of wireless links and straggling edge nodes, we have proposed a novel deep reinforcement learning, called deep dueling, to obtain the optimal code and scheduling policy for each learning task. Extensive simulation results have then demonstrated that our proposed solution can significantly improve the system performance by not only obtaining the optimal MDS code but also finding the best edge nodes to serve each learning task. One of the potential research directions from this work is to deploy multiple virtual machines at each edge node to serve various learning tasks simultaneously.
\appendices
\section{The proof of Theorem~\ref{theo:limitexists}}
\label{appendix:limitexist}
We first prove that the underlying Markov chain in this paper is irreducible. In other words, from any state, the process can always move to any other states after a finite number of steps. Recall that the system state is defined as the state of the queue $q$, the task size $f$, and the state of all edge nodes in the system $\{e_1, \ldots, e_j, \ldots, e_N\}$. At each time slot, a learning task arrives at the system with probability $\mu$. Thus, there always exists a probability that the queue state moves from $q$ to $q'= q + 1$. Moreover, a learning task will be removed from the queue if it is successfully served. In this case, the queue state moves from $q$ to $q'= q-1$. The task size is a random value. As such, it can take any positive values. Alternatively, edge node $E_j$ is available (i.e., $e_j = 1$) when it does not serve any learning task. In contrast, edge node is unavailable (i.e., $e_j = 0$) if it is serving another learning task. As a result, edge nodes can always move from the available state to unavailable state. Thus, the underlying Markov chain can move from a given state to any other states after a finite number of steps. As such, the average long-term reward $\mathcal{R}(\pi)$ is well defined and does not depend on the initial state for every $\pi$~\cite{filar2012competitive}.

\end{document}